\newcommand{\ifConferenceVersion}{\iffalse}
\newcommand{\ifJournalVersion}{\iftrue}
\newcommand{\comment}[1]{}

\ifConferenceVersion
\documentclass[a4paper,english]{lipics-v2018}
\newcommand{\JournalProof}[2]{\newcommand{#1}{#2}}

\newcommand{\InConference}[1]{#1}
\newcommand{\InJournal}[1]{}
\else
\documentclass[11pt,letter]{article}

\newcommand{\JournalProof}[2]{#2}

\newcommand{\InConference}[1]{}
\newcommand{\InJournal}[1]{#1}
\fi 

\usepackage{amsmath,amsthm}
\usepackage{graphicx}
\usepackage{color}
\usepackage{booktabs}
\usepackage{multirow}
\usepackage{url}
\usepackage{hyperref}
\InConference{
\usepackage[%
bibbreaks=normal,%
paragraphs=tight,%
floats=tight,%
mathspacing=normal,%
wordspacing=normal,%
tracking=normal,%
bibnotes=normal,%
charwidths=normal,%
mathdisplays=tight,%
leading=normal,%
indent=normal,%
lists=normal,%
bibliography=normal,%
title=normal,%
sections=normal,%
margins=normal%
]{savetrees} 
}
\newcommand{\boldm}[1] {\mathversion{bold}#1\mathversion{normal}}
\usepackage[boxed,vlined,linesnumbered]{algorithm2e}

\newtheorem{proposition}{Proposition}

\InJournal{
\usepackage{fullpage}
\usepackage{authblk}
\newtheorem{theorem}{Theorem}

\newtheorem{lemma}[theorem]{Lemma}
}

\graphicspath{{./fig/}{./}}

\InConference{
	\author{Seth Gilbert}{National University of Singapore, Singapore}{seth.gilbert@comp.nus.edu.sg}{}{}
	
	\author{Nancy Lynch}{MIT, Cambridge, MA, USA}{lynch@csail.mit.edu}{}{}
	
	\author{Calvin Newport}{Georgetown University, Washington, DC, USA}{cnewport@cs.georgetown.edu}{}{}
	
	\author{Dominik Pajak}{MIT, Cambridge, MA, USA}{pajak@csail.mit.edu}{}{}
	
	\authorrunning{S. Gilbert, N. Lynch, C. Newport and D. Pajak}
	
	\Copyright{Seth Gilbert, Nancy Lynch, Calvin Newport and Dominik Pajak}
	
	\subjclass{\ccsdesc[100]{Theory of computation~Design and analysis of algorithms~Distributed algorithms}, \ccsdesc[100]{Networks~Network types~Ad hoc networks}}
	
	\keywords{radio networks; broadcast; unreliable links; distributed algorithm; robustness}
	
	\category{}
	
	\relatedversion{The full version is available at \url{https://arxiv.org/abs/1803.02216}}
	
	\supplement{}
	
	\funding{}

\EventEditors{Jiannong Cao, Faith Ellen, Luis Rodrigues, and Bernardo Ferreira}
\EventNoEds{4}
\EventLongTitle{22nd International Conference on Principles of Distributed Systems (OPODIS 2018)}
\EventShortTitle{OPODIS 2018}
\EventAcronym{OPODIS}
\EventYear{2018}
\EventDate{December 17–19, 2018}
\EventLocation{Hong Kong, China}
\EventLogo{}
\SeriesVolume{122}
\ArticleNo{0} 

}
\begin{document}
	
	\title{On~Simple~Back-Off  in Unreliable Radio Networks\footnote{Definitions and preliminary results concerning the local broadcast problem appeared in the brief announcement~\cite{ba}, published in the Proceedings of 32nd International Symposium on DIStributed Computing (DISC) 2018.}}
	\InJournal{
	\author[1]{Seth Gilbert}
	\author[2]{Nancy Lynch}
	\author[3]{Calvin Newport}
	\author[2]{Dominik Pajak}
	\affil[1]{National University of Singapore, E-mail: seth.gilbert@comp.nus.edu.sg}
	\affil[2]{Massachusetts Institute of Technology, E-mails: \{lynch, pajak\}@csail.mit.edu}
	\affil[3]{Georgetown University, E-mail: cnewport@cs.georgetown.edu} }

	\maketitle
	\newcommand{\Decay}{\textit{Decay}\xspace}
	\newcommand{\Pro}[1]{\mathbf{Pr}\!\left[\,#1\,\right]}
	\newcommand{\Ex}[1]{\mathbf{E}\!\left[\,#1\,\right]}
	\newcommand{\Single}{\ensuremath{\mathsf{Single}}\xspace}   
	\newcommand{\Silence}{\ensuremath{\mathsf{Silence}}\xspace}
	\newcommand{\procname}[0]{\texttt{Uniform}\xspace}
	\newcommand{\globalalgname}[0]{\texttt{RGB}}
	\newcommand{\DDelta}{\dot{\Delta}}
	\newcommand{\etal}{{\it et~al.}}
	\newcommand{\algnameone}[0]{\texttt{RLB}}	
	\newcommand{\algnametwo}[0]{\texttt{FRLB}}	
	\newcommand{\algnamethree}[0]{\texttt{RLBC}}	
	
	\SetKwBlock{Repeat}{repeat}{}
	\SetEndCharOfAlgoLine{}
	
	\begin{abstract}
		In this paper, we study local and global broadcast in the dual graph model, which describes communication in a radio network with both reliable and unreliable links.
		Existing work proved that efficient solutions to these problems are impossible in the dual graph model under standard assumptions.
		In real networks, however, simple back-off strategies tend to perform well for solving these basic communication tasks.
		We address this apparent paradox by introducing a new set of constraints to the dual graph model that better generalize the slow/fast fading
		behavior common in real networks.
		We prove that in the context of these new constraints,
		simple back-off strategies now provide efficient solutions to local and global broadcast in the dual graph model.
		We also precisely characterize how this efficiency degrades as the new constraints are reduced down to non-existent,
		and prove new lower bounds that establish this degradation as near optimal for a large class of natural algorithms. 
		We conclude with an analysis of a more general model where we propose an enhanced back-off algorithm. 
		These results provide theoretical foundations for the practical observation that simple back-off algorithms tend to work well
		even amid the complicated link dynamics of real radio networks. 
	\end{abstract}
	
	\section{Introduction}
	In this paper, we study upper and lower bounds for efficient broadcast in the 
	dual graph radio network 
	model~\cite{ClementiMS04,KLN-podc09,KLNOR,CGKLN-jour,GHLN-disc12,Ghaffari-MS,GhaffariLynchNewport-sub,GKLN14,lynch:2015,ghaffari:2016},
	a dynamic network model that describes wireless communication over both reliable and unreliable links.
	As argued in previous studies of this setting,
	including unpredictable link behavior in theoretical wireless network models is important because
         in real world deployments radio links are often quite dynamic.

	
	\InConference{\vspace*{-2mm}}
	\subparagraph{The Back-Off Paradox.}
	Existing papers~\cite{KLNOR,GhaffariLynchNewport-sub,lynch:2015} proved that it is impossible to solve standard broadcast
	problems efficiently in the dual graph model without the addition of strong extra assumptions (see related work).
	In real radio networks, however, which suffer from the type of link dynamics abstracted by the dual graph model, 
	simple back-off strategies tend to perform quite well.
	%
	%
	These dueling realities seem to imply a dispiriting gap between theory and practice:
	basic communication tasks that are easily solved in real networks are impossible when studied
	in abstract models of these networks. 
	%
	
	{\em What explains this paradox?} This paper tackles this fundamental question.
	
	As detailed below, we focus our attention on the {\em adversary} entity that decides which unreliable links to include in the network
	topology in each round of an execution in the dual graph model.
	We introduce a new type of adversary with constraints that better generalize the dynamic behavior of real radio links. 
	We then reexamine simple back-off strategies originally introduced in the standard radio network model~\cite{Bar-YehudaGI92} (which has only reliable links),
	and prove that for reasonable parameters,
	these simple strategies {\em now do} guarantee 
	efficient communication in the dual graph model combined with our new, more realistic adversary.
	
	We also detail how this performance degrades toward the existing dual graph lower bounds as the new constraints are
	reduced toward non-existent, and prove lower bounds that establish these bounds to be near tight for a large and natural
	class of back-off strategies.
	Finally, we perform investigations of even more general (and therefore more difficult) variations of
	this new style of adversary that continue to underscore the versatility of simple back-off strategies.
	
	We argue that these results help resolve the back-off paradox described above. 
	When unpredictable link behavior is modeled properly,
	predictable algorithms prove to work surprisingly well.
	\InConference{\vspace*{-2mm}}	
	\subparagraph{The Dual Graph Model.}
	The dual graph model describes a radio network topology with two graphs, $G=(V,E)$ and $G'=(V,E')$,
	where
	$E \subseteq E'$,  $V$ corresponds to the wireless devices, $E$ corresponds to reliable (high quality)
	links, and $E' \setminus E$ corresponds to unreliable (quality varies over time) links.
	In each round, all edges from $E$ are included in the network topology.
	Also included is an additional subset of edges from $E' \setminus E$, chosen by an {\em adversary}.
	This subset can change from round to round.
	Once the topology is set for the round, the model implements the standard communication rules from the classical
	radio network model: a node $u$ receives a message broadcast by its neighbor $v$ in the topology if and only
	if $u$ decides to receive and $v$ is its only neighbor broadcasting in the round.
	
	We emphasize that the abstract models used in the sizable literature studying distributed algorithms in wireless settings 
	do not claim to provide high fidelity representations of real world radio signal communication.
	They instead each capture core dynamics of this setting, enabling the investigation of fundamental algorithmic questions.
	The well-studied radio network model, for example, provides a simple but instructive abstraction of message loss due to collision.
	The dual graph model generalizes this abstraction to also include network topology dynamics. Studying the gaps between
	these two models provides insight into the hardness induced by the types of link quality changes common in real wireless networks.
		
	\InConference{\vspace*{-2mm}}	
	\subparagraph{The Fading Adversary.}
	Existing studies of the dual graph model focused mainly on the information about the algorithm known to the model adversary
	when it makes its edge choices.
	In this paper, we place additional constraints on how these choices are generated.
	
	In more detail, 
	in each round, the adversary independently draws the set of edges from $E' \setminus E$ to add to the topology
	from some probability distribution defined over this set.
	We do not constrain the properties of the distributions selected by the adversary. Indeed, it is perfectly valid for the 
	adversary in a given round to use a point distribution that puts the full probability mass on a single subset,
	giving it full control over its selection for the round. We also assume the algorithm executing in the model has no
	advance knowledge of the distributions used by the adversary. 
	
	We do, however, constrain how often the adversary can change the distribution from which it selects these edge subsets.
	In more detail, we parameterize the model with a {\em stability factor}, $\tau \geq 1$,
	and restrict the adversary to changing the distribution it uses at most once every $\tau$ rounds.
	For $\tau =1$, the adversary can change the distribution in every round,
	and is therefore effectively unconstrained and behaves the same as in the existing dual graph studies.
	On the other extreme, for $\tau=\infty$, the adversary is now quite constrained in that it must draw edges independently from
	the same distribution for the entire execution. 
	As detailed below, we find $\tau \approx \log{\Delta}$, for local neighborhood size $\Delta$,
	to be a key threshold after which efficient communication becomes tractable.
	
	Notice, these constraints do not prevent the adversary from inducing large amounts of changes to the network topology
	from round to round. For non-trivial $\tau$ values, however, they do require changes that are nearby in time to share
	some underlying stochastic structure. This property is inspired by the general way wireless network engineers think about unreliability
	in radio links. 
	In their analytical models of link behavior (used, for example, to analyze modulation or rate selection schemes, or to model signal propagation in simulation),
	engineers often assume that in the short term, 
	changes to link quality come from sources like noise and multi-path effects, which can be approximated by independent draws from
	an underlying distribution (Gaussian distributions are common choices for this purpose).
	Long term changes, by contrast, can come
          from modifications to the network environment itself,
	such as devices moving, which do not necessarily have an obvious stochastic structure,
	but unfold at a slower rate than short term fluctuations.
	
	In our model, the distribution used in a given round captures short term changes,
	while the adversary's arbitrary (but rate-limited) changes to these distributions over time capture  long term changes.
	Because these general types of changes are sometimes labeled {\em short/fast fading} in the systems literature (e.g.,~\cite{fading}),
	we call our new adversary a {\em fading adversary}.
	

	\begin{table}\centering
		\begin{tabular}{lllll}\toprule[1pt]
			\textbf{Problem}& \textbf{Time} & \textbf{Prob.}  & \textbf{Remarks} & \textbf{Ref.} \\\midrule[1pt]
			\multirow{4}{*}{Local broadcast} & $O\left(\frac{\Delta^{1/\bar{\tau}} \cdot \bar{\tau}^2}{\log\Delta}\cdot  \log{(1/\epsilon)}\right)$ & $1-\epsilon$ & $\bar{\tau} = \min\{\tau,\log\Delta\}$ & Thm~\ref{thm:mult_receivers}\\ 
			& $\Omega\left( \frac{\Delta^{1/\tau} \tau}{\log{\Delta}}\right)$ & $\frac12$ & $\tau \in O(\log\Delta)$ & Thm~\ref{thm:lower} \\[1mm]
			& $\Omega\left( \frac{\Delta^{1/\tau} \tau^2}{\log{\Delta}}\right)$ & $\frac12$ & $\tau \in O(\log\Delta/\log\log\Delta)$ & Thm~\ref{thm:lower_2} 
			\\\midrule[1pt]
			\multirow{4}{*}{Global broadcast} & $O\left((D+ \log (n/\epsilon))\cdot  \frac{\Delta^{1/\bar{\tau}} \bar{\tau}^2}{\log\Delta}\right)$ & $1-\epsilon$ & $\bar{\tau} = \min\{\tau,\log\Delta\}$  & Thm~\ref{thm:global_upper}\\
			& $\Omega\left(D\cdot  \frac{\Delta^{1/\tau} \tau}{\log\Delta}\right)$ & $\frac12$ & $\tau \in O(\log\Delta)$ & Thm~\ref{thm:global_lower}\\		
			& $\Omega\left(D\cdot  \frac{\Delta^{1/\tau} \tau^2}{\log\Delta}\right)$ & $\frac12$  & $\tau \in O(\log\Delta/\log\log\Delta)$ & Thm~\ref{thm:global_lower}\\
			\bottomrule[1pt]
		\end{tabular}
		\caption{A summary of the upper and lower bounds proved in this paper, along with pointers to the corresponding theorems.
			In the following, $n$ is the network size, $\Delta \leq n$
			is an upper bound on local neighborhood size,  $D$ is the (reliable link) network diameter, 
			and $\tau$ is the stability factor constraining the adversary.}
		\label{fig:results}
	\end{table}

	\subparagraph{Our Results and Related Work.}
	In this paper, we study both {\em local} and {\em global} broadcast.
	The local version of this problems assumes some subset of devices in a dual graph network are provided broadcast messages.
	The problem is solved once each receiver that neighbors a broadcaster in $E$ receives at least one message.
	The global version assumes a single broadcaster starts with a message that it must disseminate to the entire network.
	Below we summarize the relevant related work on these problems, and the new bounds proved in this paper. We conclude
	with a discussion of the key ideas behind these new results.
	
	{\em Related Work.} In the standard radio network model, which is equivalent to the dual graph model with $E=E'$, 
	Bar-Yehuda et~al.~\cite{Bar-YehudaGI92} demonstrate that a simple randomized back-off strategy called \Decay
	solves  local broadcast in $O(\log^2{n})$ rounds and global broadcast in $O(D\log{n} + \log^2{n})$ rounds,
	where $n=|V|$ is the network size and $D$ is the diameter of $G$.
	Both results hold with high probability in $n$, and were subsequently proved 
	to be optimal or near optimal\footnote{The broadcast algorithm from~\cite{Bar-YehudaGI92} requires
		$O(D\log{n}+\log^2{n})$ rounds, whereas the corresponding lower bound is $\Omega(D\log{(n/D)} + \log^2{n})$.
		This gap was subsequently closed by a tighter analysis of a natural variation of the simple \Decay
		strategy used in~\cite{Bar-YehudaGI92} }~\cite{alon:1991,km, newport:2014b}.
	
	In~\cite{KLN-podc09,KLNOR}, it is proved that global broadcast (with constant diameter),
	and local broadcast require $\Omega(n)$ rounds to solve with reasonable probability
	in the dual graph model with an offline adaptive adversary controlling the unreliable edge selection,
	while~\cite{GhaffariLynchNewport-sub} proves that $\Omega(n/\log{n})$ rounds are necessary for both problems with an online adaptive adversary.
	As also proved in~\cite{GhaffariLynchNewport-sub}:
	even with the weaker oblivious adversary,
	local broadcast requires $\Omega(\sqrt{n}/\log{n})$ rounds,
	whereas global broadcast {\em can} be solved in an efficient $O(D\log{(n/D)} + \log^2{n})$ rounds,
	but only if the broadcast message is sufficiently large to contain enough shared random bits for all nodes to use throughout the execution.
	In~\cite{lynch:2015},
	an efficient algorithm for local broadcast with an oblivious adversary is provided given the assumption of geographic constraints on the dual graphs, enabling complicated clustering strategies that allow nearby devices to coordinate randomness.
	
	{\em New Results.}
	In this paper, we turn our attention to local and global broadcast in the dual graph model with a fading adversary constrained by some
	stability factor $\tau$ (unknown to the algorithm). 
	We start by considering upper bounds for a simple back-off style strategy inspired by the {\em decay} routine from~\cite{Bar-YehudaGI92}.
	This routine has broadcasters simply cycle through a fixed set of broadcast probabilities in a synchronized manner (all broadcasters use the same probability in the same round). 
	We prove that this strategy solves local broadcast with probability at least $1-\epsilon$,
	in $O\left(\frac{\Delta^{1/\bar{\tau}} \cdot \bar{\tau}^2}{\log\Delta}\cdot  \log{(1/\epsilon)}\right)$ rounds,
	where $\Delta$ is an upper bound on local neighborhood size, and $\bar{\tau} = \min\{\tau,\log\Delta\}$.
	
	Notice, for $\tau \geq \log{\Delta}$ this bound simplifies to $O(\log{\Delta}\log{(1/\epsilon)})$, matching the optimal results from the standard radio network 
	model.\footnote{To make it match exactly, set $\Delta=n$ and $\epsilon = 1/n$, as is often assumed in this prior work.}
	This performance, however, degrades toward the polynomial lower bounds from the existing dual graph literature
	as $\tau$ reduces from
	$\log{\Delta}$ toward a minimum value of $1$. 
	We show this degradation to be near optimal by proving that {\em any}
	local broadcast algorithm that uses a fixed sequence of broadcast probabilities
	requires $\Omega(\Delta^{1/\tau} \tau/\log{\Delta})$ rounds to solve the problem with probability $1/2$ for a given $\tau$.
	For $\tau \in O(\log\Delta/\log\log\Delta)$ , we refine this bound further to $\Omega(\Delta^{1/\tau} \tau^2 /\log{\Delta})$, matching
	our upper bound within constant factors. 
	
	We next turn our attention to global broadcast. We consider a straightforward global broadcast algorithm that uses our local broadcast strategy as a subroutine. We prove that this algorithm solves global broadcast with probability at least $1-\epsilon$, in 
	$O(D+ \log (n/\epsilon))\cdot  \Delta^{1/\bar{\tau}} \bar{\tau}^2/\log\Delta)$ rounds,
	where $D$ is the diameter of $G$, and  $\bar{\tau} = \min\{\tau,\log\Delta\}$.
	Notice, for $\tau \geq \log{\Delta}$ this bound reduces to $O(D\log{\Delta} + \log{\Delta}\log{(1/\epsilon)})$, matching
	the near optimal result from the standard radio network model.
	As with local broadcast, we also prove the degradation of this performance as $\tau$ shrinks to be near optimal.
	(See Table~\ref{fig:results} for a summary of these results and pointers to where they are proved in this paper.)
	
	
	Finally we consider the generalized model when we allow correlation between the distributions selected by the adversary within a given stable period of $\tau$ rounds. It turns out that in the case of arbitrary correlations any simple algorithm needs time $\Omega(\sqrt{\Delta}/l)$ if it uses only cycles of length $l$. In particular any our previous algorithms would require time $\Omega(\sqrt{\Delta}/\log \Delta)$ in the model with arbitrary correlations. The adversary construction in this lower bound requires large changes in the degree of a node in successive steps. Such changes are unlikely in real networks thus we propose a restricted version of the adversary. We assume that the expected change in the degree of any node can be at most $\Delta^{1/(\bar{\tau}(1- o(1))}$. With such restriction it is again possible to propose a simple, but slightly enhanced, back-off strategy (with a short cycle of probabilities) that works efficiently in time $O\left(\Delta^{1/\bar{\tau}} \cdot \bar{\tau}\cdot \log{(1/\epsilon)}\right)$. 
	
	{\em Technique Discussion.} Simple back-off strategies can be understood as experimenting 
	with different {\em guesses} at the amount of contention afflicting a given receiver.
	If the network topology is static, this contention is fixed, therefore so is the {\em right} guess.
	A simple strategy cycling through a reasonable set of guesses will soon arrive at this right guess---giving the
	message a good chance of propagating.
	
	The existing lower bounds in the dual graph setting deploy
	an adversary that changes the topology in each round to specifically thwart that round's guess.
	In this way, the algorithm never has the right guess for the current round so its probability of progress is diminished. 
	The fading adversary, by contrast, is prevented from adopting this degenerate behavior because it is required to stick with the same
	distribution for $\tau$ consecutive rounds. An important analysis at the core of our upper bounds
	reveals that any fixed distribution 
	will be associated with a right guess defined with respect to the details of that distribution.
	If $\tau$ is sufficiently large, our algorithms are able to experiment with enough guesses to hit on this right guess before the adversary 
	is able to change the distribution.
	
	More generally speaking, the difficulty of broadcast in the previous dual graph studies was {\em not} due to the ability
	of the topology to change dramatically from round to round (which can happen in practice),
	but instead due to the model's ability to precisely tune these changes to thwart the algorithm (a behavior that is hard to motivate).
	The dual graph model with the fading adversary preserves the former (realistic) behavior while minimizing the latter (unrealistic) behavior.

\InConference{\vspace*{-2mm}}
\section{Model and Problem}
\InConference{\vspace*{-1mm}}
We study the dual graph model of unreliable radio networks.
This model describes the network topology with two graphs $G=(V,E)$ and $G' = (V,E')$, where $E \subseteq E'$.
The $n=|V|$ vertices in $V$ correspond to the wireless devices in the network, which we call {\em nodes} in the following.
The edge in $E$ describe reliable links (which maintain a consistently high quality), while the edges in $E' \setminus E$
describe unreliable links (which have quality that can vary over time). 
For a given dual graph, we use $\Delta$ to describe the maximum degree in $G'$,
and $D$ to describe the diameter of $G$.

Time proceeds in synchronous rounds that we label $1,2,3...$
For each round $r\geq 1$,
the network topology is described by $G_r = (V,E_r)$,
where $E_r$ contains all edges in $E$ plus a subset of the edges in $E' \setminus E$.
The subset of edges from $E' \setminus E$ are selected by an {\em adversary}.
The graph $G_r$ can be interpreted as describing the high quality links during round $r$.
That is, if $\{u,v\} \in E_r$,
this mean the link between $u$ and $v$ is strong enough
that $u$ could deliver a message to $v$,
or garble another message being sent to $v$ at the same time.

With the topology $G_r$ established for the round,
behavior proceeds as in the standard radio network model.
That is, each node $u\in V$ can decide to transmit or receive.
If $u$ transmits, it learns nothing about other messages transmitted in the round (i.e., the radios are half-duplex).
If $u$ receives and exactly one neighbor $v$ of $u$ in $E_r$ transmits,
then $u$ receives $v$'s message.
If $u$ receives and two or more neighbors in $E_r$ transmit, $u$ receives nothing as the messages are lost due to collision.
If $u$ receives and no neighbor transmits, $u$ also receives nothing.
We assume $u$ does not have collision detection, meaning it cannot distinguish between these last two cases.
\InConference{\vspace*{-2mm}}
\subparagraph{The Fading Adversary.}
A key assumption in studying the dual graph model are the constraints placed on the adversary that selects the unreliable
edges to include in the network topology in each round.
In this paper, we study a new set of constraints inspired by real network behavior.
In more detail, we parameterize the adversary with a {\em stability factor} that we
represent with an integer $\tau \geq 1$.
In each round, the adversary must draw the subset of edges (if any) from $E' \setminus E$ to include
in the topology from a distribution defined over these edges.
The adversary selects which distributions it uses.
Indeed, we assume it is {\em adaptive} in the sense that it can wait until the beginning of a given round
before deciding the distribution it will use in that round, basing its decision on the history of the nodes' transmit/receive behavior
up to this point, including the previous messages they send, but not including knowledge of the nodes' private random bits.

The adversary is constrained, however, in that it can change this distribution at most once every $\tau$ rounds.
On one extreme, if $\tau=1$, it can change the distribution in every round and is effectively unconstrained in its choices.
On the other other extreme, if $\tau=\infty$, it must stick with the same distribution for every round.
For most of this paper, we assume the draws from these distributions are independent in each round. 
Toward the end, however, we briefly discuss what happens when we generalize the model to allow more correlations.

As detailed in the introduction, because these constraints roughly approximate the fast/slow fading behavior
common in the study of real wireless networks, we call a dual graph adversary constrained in this manner
a {\em fading adversary}.


\InConference{\vspace*{-2mm}}
\subparagraph{Problem.}
In this paper, we study both the {\em local} and {\em global} broadcast problems.
The local broadcast problem assumes a set $B \subseteq V$ of nodes are provided with a message to broadcast. Each node can receive a unique message.
Let $R\subseteq V$ be the set of nodes in $V$ that neighbor at least one node in $B$ in $E$.
The problem is solved once every node in $R$ has received at least one message from a node in $B$.
We assume all nodes in $B$ start the execution during round $1$, but do not require that $B$ and $R$ are disjoint (i.e., broadcasters can also be receivers).
The global broadcast problem, by contrast, assumes a single source node in $V$ is provided a broadcast message during round $1$.
The problem is solved once all nodes have received this message. 
Notice, local broadcast solutions are often used as subroutines to help solve global broadcast.
\InConference{\vspace*{-2mm}}
\subparagraph{Uniform Algorithms.}
The broadcast upper and lower bounds we study in this paper focus on {\em uniform algorithms}, which require nodes to make their probabilistic 
transmission decisions according to a predetermined sequence of broadcast probabilities
that we express as a repeating cycle, $(p_1,p_2, ..., p_k)$ of $k$ probabilities in synchrony.
 In studying global broadcast, we assume that on first receiving a message,
 a node can wait to start making probabilistic transmission decisions until the cycle resets.
We assume these probabilities can depend on $n$, $\Delta$ and $\tau$ (or worst-case bounds on these values).

In uniform algorithms in the model with fading adversary an important parameter of any node $v$ is its {\em effective degree} in step $t$ denoted by $d_t(v)$ and defined as the number of nodes $w$ such that $(v,w) \in E_t$ and $w$ has a message to transmit (i.e., will participate in step $t$). 

As mentioned in the introduction,
uniform algorithms, such as the {\em decay} strategy from~\cite{Bar-YehudaGI92},
solve local and global broadcast with optimal efficiency in the standard radio network model.
A major focus of this paper is to prove that they work well in the dual graph model as well,
if we assume a fading adversary with a reasonable stability factor.

The fact that our lower bounds assume the algorithms are uniform technically weaken the results,
as there might be non-uniform strategies that work better. 
In the standard radio network model, however, this does not prove to be the case:
uniform algorithms for local and global broadcast match lower bounds that hold for all algorithms (c.f., discussion in~\cite{newport:2014b}).

\InConference{\vspace*{-3mm}}
\section{Local broadcast}
\label{sec:local}
\InConference{\vspace*{-1mm}}
We begin by studying upper and lower bounds for the local broadcast problem.
Our upper bound performs efficiently once the stability factor $\tau$ reaches a threshold of $\log{\Delta}$.
As $\tau$ decreases toward a minimum value of $1$, this efficiency degrades rapidly.
Our lower bounds capture that this degradation for small $\tau$ is unavoidable for uniform algorithms. In the following we use the notation $\bar{\tau} =  \min\{\tau,\lceil\log\Delta\rceil\}$. By $\log n$ we will always denote logarithm at base $2$ and by $\ln n$ the natural logarithm.
\InConference{\vspace*{-3mm}}
\subsection{Upper Bound}	
\InConference{\vspace*{-1mm}}
All uniform local broadcast algorithms behave in the same manner: the nodes in $B$ repeatedly broadcast according to 
some fixed cycle of $k$ broadcast probabilities. We formalize this strategy with algorithm $\algnameone$ (Robust Local Broadcast) described below (we break out \procname into its own procedure as we later
use it in our improved $\algnametwo$ local broadcast algorithm as well):
\InJournal{\hspace*{-3mm}}
\InConference{\hspace*{-32mm}}
\begin{minipage}[t]{0.645\textwidth}
	\vspace{0pt}  
\begin{procedure}[H]
	\TitleOfAlgo{$\procname(k,p_1,p_2,\dots,p_k)$}	
	\For{$i = 1,2,\dots,k$}{
		\uIf{has message}{
			with probability $p_i$ \texttt{Transmit} otherwise \texttt{Listen} 
		}\lElse{\texttt{Listen} \tcp*[h]{without a message always listen}}}
	\label{alg:procedure}
\end{procedure}
\end{minipage}%
\hspace*{0mm}
\begin{minipage}[t]{0.38\textwidth}
	\begin{algorithm}[H]
	\TitleOfAlgo{\algnameone$(r,\bar{\tau})$}
	\lFor{$i\leftarrow 1$ \KwTo $\bar{\tau}$}{$p_i \gets \Delta^{-i/\bar{\tau}}$}
	\Repeat($r$ times){
		\procname$(\bar{\tau},p_1,p_2,\dots,p_{\bar{\tau}})$ }
\end{algorithm}
\end{minipage}

Before we prove the complexity of $\algnameone$ we will show two useful properties of any uniform algorithm. Let $R_t^{(v)}$ denote the event  that node $v$ receives a message from some neighbor in step $t$.

\begin{lemma}
	\label{lem:prosing}
	For any uniform algorithm and any node $v$ and step $t$ if $d_t(v) > 0$ and the algorithm uses in step $t$ probability $p \leq 1/2$, then $\Pro{R_{t}^{(v)}} \geq \frac{p \cdot d_t(v)}{(2 e)^{p \cdot d_t(v)}}$.
\end{lemma}
\begin{proof}
	For this to happen exactly one among $d_t(v)$ neighbors of $v$ has to transmit and $v$ must not transmit. Node $v$ does not transmit with probability $1-p$ if it has the message and clearly with probability $1$ if it has the message. Denote by $\alpha = p \cdot d_t(v)$. We have
	\begin{align*}
	\Pro{R_{t}^{(v)}} &\geq  pd_t(v) \cdot (1-p)^{d_t(v)} = \alpha \cdot \left(1-\frac{\alpha}{d_t(v)}\right)^{d_t(v)} \\&  = \alpha \left(\left(1-\frac{\alpha}{d_i(v)}\right)^{d_t(v)/\alpha - 1} \cdot (1-p) \right)^{\alpha}  \geq \alpha(e^{-1} (1-p))^{\alpha} \geq  \frac{\alpha}{(2e)^{\alpha}}.
	\end{align*}
%
\end{proof}

\begin{lemma}
	\label{lem:interval}
	For any uniform algorithm, node $v$ and step $t$ if $d_t(v) > 0$:
	\[
	\Pro{R_{t}^{(v)} \mid d_t(v) \in [d_1,d_2]} \geq \min\left\{\Pro{R_{t}^{(v)} \mid d_t(v) = d_1},\Pro{R_{t}^{(v)} \mid d_t(v) = d_2}\right\}.
	\]
\end{lemma}
\begin{proof}
	If the algorithm uses probability $p$ in step $t$ then $\Pro{R_{t}^{(v)}} = p d_t(v) (1-p)^{d_t(v)}$. Seeing this expression as a function of $d_t(v)$ we can compute the derivative and obtain that this function has a single maximum in $d_t(v) = 1/(\ln(1/(1-p)))$. Hence if we restrict $d_t(v)$ to be within a certain interval, then value of the function is lower bounded by the minimum at the endpoints of the interval. 
\end{proof}

Our upper bound analysis leverages the following useful lemma which can be shown by induction on $n$ (the left side is also known as the Weierstrass Product Inequality):
\begin{lemma}
	\label{lem:wpi}
	For any $x_1,x_2,\dots,x_n$ such that $0\leq x_i \leq 1$:
	\[
	1 - \sum_{i=1}^n x_i \leq \prod_{i=1}^n\left(1- x_i\right) \leq 1- \sum_{i=1}^{n}x_i + \sum_{1\leq i<j\leq n}x_i x_j.
	\]
\end{lemma}

To begin our analysis, we focus on the behavior of our algorithm with respect to a single receiver
when we use the transmit probability sequence $p_1, p_2, ..., p_{\bar{\tau}}$,
where $\bar{\tau} = \min\{\tau,\lceil\log\Delta\rceil\}$, and $p_i = \Delta^{-i/\bar{\tau}}$.

\begin{lemma}
	\label{lem:upper_1}
	Fix any receiver $u \in R$ and error bound $\epsilon > 0$. It follows:
	\algnameone$(2 \lceil \ln (1/\epsilon) \rceil\cdot \lceil 4e\cdot \Delta^{1/\bar{\tau}} \rceil,\bar{\tau})$ delivers a message to $u$ with probability at least $1-\epsilon$ in time $O(\Delta^{1/\bar{\tau}} \bar{\tau} \log(1/\epsilon))$.
\end{lemma}	
\InConference{
\begin{proof}
	It is sufficient to prove the claim for $\tau \leq \log\Delta$. For $\tau > \log\Delta$ we use the algorithm for $\tau = \log\Delta$.  Note that any algorithm that is correct for some $\tau$ must also work for any larger $\tau$ because the adversary may not choose to change the distribution as frequently as it is permitted to. In the case where $\tau \leq \log\Delta$ we get that $\Delta^{1/\tau} \geq 2$.
	
	We want to show that if the nodes from $N_{u}\cap B$ execute procedure $\procname(\tau,p_1,\dots,p_{\tau})$ twice, then $u$ receives some message with probability at least $\log\Delta/(2e \Delta^{1/\tau}\tau)$. Every time we execute $\procname$ twice, we have a total of $2\tau$ consecutive time slots out of which, by the definition of our model, at least $\tau$ consecutive slots have the same distribution of the additional edges and moreover stations try all the probabilities $p_1,p_2,\dots,p_{\tau}$ (not necessarily in this order). Let $T$ denote the set of these $\tau$ time slots and for $i= 1,2,\dots,\tau$ let $t_i \in T$ be the step in which probability $p_i$ is used. We also denote the distribution used in steps from set $T$ by  $\mathcal{E^{(T)}}$. Hence we can denote the edges between $u$ and its neighbors that have some message by $E_{part} = \{(u,b): b\in B\}\cap E'$. We know that the edge sets are chosen independently from the same distribution: $E_t \sim \mathcal{E^{(T)}}$ for $t \in T$. Let us denote by $X_t = |E_t\cap E_{part}|$ the random variable being the number of neighbors that are connected to $u$ in step $t$ and belong to $B$. For each $i$ from $1$ to $\tau$ we define
	$
	q_i = \Pro{\Delta^{(i-1)/\tau} < X_t \leq \Delta^{i/\tau}},
	$
	for any $t\in T$. Observe that probabilities $q_i$ do not depend on $t$ during the considered $\tau$ rounds. Moreover since $u \in R$ then $u$ is connected via a reliable edge to at least one node in $B$, thus $E\cap E_{part} \neq \emptyset$, hence $\Pro{X_t = 0} = 0$ thus:
	\begin{equation}
	\label{eqn:sum_one}
	\sum_{i = 1}^{\tau} q_i = 1,
	\end{equation}
	Let $S_i$ denote the indicator random variable being $1$ if in $t_i$-th round if exactly one neighbor of $u$ transmits and $u$ is not transmitting in round $t$ and $0$ otherwise. Clearly if $S_i = 1$ in some round $t$, then $u$ receives some message in round $t$. Then we would like to show for each $i = 1,2,\dots,\tau$ that:
	\begin{equation}
	\label{eqn:s_i}
	\Pro{S_i = 1}\geq  \frac{q_i}{2e\Delta^{1/\tau}}.
	\end{equation}
	In $t_i$-th slot the transmission probability is $p_i = \Delta^{-i/\tau}$ and the transmission choices done by the stations are independent from the choice of edges $E_{t_i}$ active in round $t_i$. Note that $u$ might also belong $R$ and try to transmit. But since $p_i \leq 1/2$ then $u$ is not transmitting with probability at least $1/2$. If $Q_{i}$ denotes the event that $\Delta^{(i-1)/\tau} < X_{t_i} \leq \Delta^{i/\tau}$ then:
	\begin{align*}
	\Pro{S_{i} = 1} &\geq \Pro{S_i = 1| Q_{i}}\cdot \Pro{Q_{i}} \\
	&\geq  p_{i} (\Delta^{(i-1)/\tau} + 1) \cdot \left(1-p_i\right)^{\Delta^{(i-1)/\tau}} \cdot \frac12 \cdot q_{i}\\
	&\geq  p_{i} \Delta^{(i-1)/\tau} \cdot \left(1-p_i\right)^{\Delta^{i/\tau} - 1} \cdot \frac12 \cdot q_{i}\\	
	& \geq \Delta^{-1/\tau}\cdot \left(1-\frac{1}{\Delta^{i/\tau}}\right)^{\Delta^{i/\tau} - 1} \cdot \frac{q_i}{2} \geq \frac{q_i}{2e \Delta^{1/\tau}},
	\end{align*}
	because inequality $(1-1/x)^{x-1} \geq e^{-1}$ holds for all $x > 0$. Since the edge sets are chosen independently in each step and the random choices of the stations whether to transmit or not are also independent from each other we have:
	\begin{alignat*}{3}
	\Pro{\bigwedge_{i=1}^{\tau} (S_i = 0)} &= \prod_{i = 1}^{\tau}\Pro{S_i = 0} \leq \prod_{i = 1}^{\tau} \left(1- \frac{q_i}{2e \Delta^{1/\tau}}\right)& \quad \text{by Equation~\eqref{eqn:s_i}}\\
	&\leq 1 - \sum_{i = 1}^{\tau}\frac{q_i}{2e \Delta^{1/\tau}} + \sum_{1\leq i < j\leq \tau} \frac{q_i q_{j}}{4e^2 \Delta^{2/\tau} } & \text{by Lemma~\ref{lem:wpi}}\\
	& \leq 1 - \frac{\sum_{i = 1}^{\tau}q_i}{2e \Delta^{1/\tau}} + \frac{ \left(\sum_{i =1}^{\tau} q_i\right)^2}{4e^2 \Delta^{2/\tau} } & \\
	&\leq  1 - \frac{1}{2e \Delta^{1/\tau}} + \frac{1}{4e^2 \Delta^{2/\tau}} \leq 1 - \frac{1}{4e \Delta^{1/\tau}} &\text{by Equation~\eqref{eqn:sum_one}} 
	\end{alignat*}
	Hence if we execute the procedure for $ 2\tau\lceil \ln (1/\epsilon) \rceil\cdot \lceil 4e\cdot \Delta^{1/\tau} \rceil$ time steps, we have at least $\lceil \ln (1/\epsilon) \rceil\cdot \lceil 4e\cdot \Delta^{1/\tau} \rceil$ sequences of $\tau$ consecutive time steps in which the distribution over the unreliable edges is the same and the algorithm tries all the probabilities $\{p_1,p_2,\dots,p_{\tau}\}$. Each of these procedures fails independently with probability at most $1- 1/(4e \Delta^{1/\tau})$ hence the probability that all the procedures fail is at most:
	$
	\left(1-\frac{1}{4e \Delta^{1/\tau}}\right)^{\lceil \ln (1/\epsilon) \rceil\cdot \lceil 4e\Delta^{1/\tau} \rceil} \leq e^{-\lceil \ln (1/\epsilon)\rceil} < \epsilon
	$
\end{proof}}

\InJournal{
\begin{proof}
	It is sufficient to prove the claim for $\tau \leq \log\Delta$. For $\tau > \log\Delta$ we use the algorithm for $\tau = \log\Delta$.  Note that any algorithm that is correct for some $\tau$ must also work for any larger $\tau$ because the adversary may not choose to change the distribution as frequently as it is permitted to. In the case where $\tau \leq \log\Delta$ we get that $\Delta^{1/\tau} \geq 2$.
	
	 We want to show that if the nodes from $N_{u}\cap B$ execute the procedure $\procname(\tau,p_1,\dots,p_{\tau})$ twice, then $u$ receives some message with probability at least $\log\Delta/(2e \Delta^{1/\tau}\tau)$. Every time we execute $\procname$ twice, we have a total of $2\tau$ consecutive time slots out of which, by the definition of our model, at least $\tau$ consecutive slots have the same distribution of the additional edges and moreover stations try all the probabilities $p_1,p_2,\dots,p_{\tau}$ (not necessarily in this order). Let $T$ denote the set of these $\tau$ time slots and for $i= 1,2,\dots,\tau$ let $t_i \in T$ be the step in which probability $p_i$ is used. We also denote the distribution used in steps from set $T$ by  $\mathcal{E^{(T)}}$. Hence we can denote the edges between $u$ and its neighbors that have some message by $E_{part} = \{(u,b): b\in B\}\cap E'$. We know that the edge sets are chosen independently from the same distribution: $E_t \sim \mathcal{E^{(T)}}$ for $t \in T$. Let us denote by $X_t = |E_t\cap E_{part}|$ the random variable being the number of neighbors that are connected to $u$ in step $t$ and belong to $B$. For each $i$ form $1$ to $\tau$ we define
	$
	q_i = \Pro{\Delta^{(i-1)/\tau} < X_t \leq \Delta^{i/\tau}},
	$
	for any $t\in T$. Observe that probabilities $q_i$ do not depend on $t$ during the considered $\tau$ rounds. Moreover since $u \in R$ then $u$ is connected via a reliable edge to at least one node in $B$, thus $E_t \cap E_{part} \neq \emptyset$, hence $\Pro{X_t = 0} = 0$ thus:
	\begin{equation}
	\label{eqn:sum_one}
	\sum_{i = 1}^{\tau} q_i = 1.
	\end{equation}
	We would like to lower bound the probability that $v$ receives a message in step $t_i$ for $i = 1,2\dots,\tau$:
	\begin{equation}
	\label{eqn:s_i}
	\Pro{R_{t_i}^{(v)}}\geq  \frac{q_i}{2e\Delta^{1/\tau}}.
	\end{equation}
	In $t_i$-th slot the transmission probability is $p_i = \Delta^{-i/\tau}$ and the transmission choices done by the stations are independent from the choice of edges $E_{t_i}$ active in round $t_i$. Let $Q_i$ denote the event that $\Delta^{(i-1)/\tau} < X_{t_i} \leq \Delta^{i/\tau}$. We have $p_i \leq 1/2$ hence we can use Lemma~\ref{lem:prosing} and~\ref{lem:interval}:
	\begin{align}\nonumber
	\Pro{R_{t_i}^{(v)}} &\geq \Pro{R_{t_i}^{(v)} = 1 \mid Q_{i}}\cdot \Pro{Q_{i}} \\\label{eqn:minimum}
	& \geq q_i\min\left\{\Pro{R_{t_i}^{(v)} = 1\mid d_{t_i}(v) = \Delta^{(i-1)/\tau} +1}, \Pro{R_{t_i}^{(v)} = 1\mid d_{t_i}(v) = \Delta^{i/\tau}}\right\}\\\nonumber
	& \geq q_i\min\left\{\frac{(\Delta^{(i-1)/\tau} +1)\Delta^{-i/\tau}}{(2e)^{(\Delta^{(i-1)/\tau} +1)\Delta^{-i/\tau}}}, \frac{1}{2e}\right\} \geq  \frac{q_i}{2e\Delta^{1/\tau}},
	\end{align}
	because $(\Delta^{(i-1)/\tau} +1)\Delta^{-i/\tau} \leq 1$. Since the edge sets are chosen independently in each step and the random choices of the stations whether to transmit or not are also independent from each other we have:
	\begin{alignat*}{3}
	\Pro{\bigwedge_{i=1}^{\tau} \left(\neg R_{t_i}^{(v)}\right)} &= \prod_{i = 1}^{\tau}\Pro{\neg R_{t_i}^{(v)}} \leq \prod_{i = 1}^{\tau} \left(1- \frac{q_i}{2e \Delta^{1/\tau}}\right)& \quad\text{by independence and Equation~\eqref{eqn:s_i}}\\
	&\leq 1 - \sum_{i = 1}^{\tau}\frac{q_i}{2e \Delta^{1/\tau}} + \sum_{1\leq i < j\leq \tau} \frac{q_i q_{j}}{4e^2 \Delta^{2/\tau} } & \text{by Lemma~\ref{lem:wpi}}\\
	& \leq 1 - \frac{\sum_{i = 1}^{\tau}q_i}{2e \Delta^{1/\tau}} + \frac{ \left(\sum_{i =1}^{\tau} q_i\right)^2}{4e^2 \Delta^{2/\tau} } & \\
	&\leq  1 - \frac{1}{2e \Delta^{1/\tau}} + \frac{1}{4e^2 \Delta^{2/\tau}} \leq 1 - \frac{1}{4e \Delta^{1/\tau}} &\text{by Equation~\eqref{eqn:sum_one}}.	
	\end{alignat*}
	Hence if we execute the procedure for $ 2\tau\lceil \ln (1/\epsilon) \rceil\cdot \lceil 4e\cdot \Delta^{1/\tau} \rceil$ time steps, we have at least $\lceil \ln (1/\epsilon) \rceil\cdot \lceil 4e\cdot \Delta^{1/\tau} \rceil$ sequences of $\tau$ consecutive time steps in which the distribution over the unreliable edges is the same and the algorithm tries all the probabilities $\{p_1,p_2,\dots,p_{\tau}\}$. Each of these procedures fails independently with probability at most $1- 1/(4e \Delta^{1/\tau})$ hence the probability that all the procedures fail is at most:
	$
	\left(1-\frac{1}{4e \Delta^{1/\tau}}\right)^{\lceil \ln (1/\epsilon) \rceil\cdot \lceil 4e\Delta^{1/\tau} \rceil} \leq e^{-\lceil \ln (1/\epsilon)\rceil} < \epsilon
	$
\end{proof}}

On closer inspection of the analysis of Lemma~\ref{lem:upper_1},
it becomes clear that if we tweak slightly the probabilities used in our algorithm, we require fewer iterations.
In more detail, the probability of a successful transmission in the case where each of the $x$ transmitters broadcasts independently with probability $\alpha/x$ is approximately $\alpha / (2e)^{\alpha}$. In the previous algorithm we were transmitting in successive steps with probabilities $\Delta^{-1/\tau}, \Delta^{-2/\tau},\dots$. Thus if $x = 1$ we would get in $i$-th step $\alpha = \Delta^{-i/\tau}$ and approximately the sum of probabilities of success in $\tau$ consecutive steps would be $\Delta^{-1/\tau}$. The formula $\alpha / (2e)^{-\alpha}$ shows that the success probability depends on $\alpha$ linearly if $\alpha < 1$ (``too small" probability) and depends exponentially on $\alpha$ if $\alpha > 1$ (``too large" probability). In the previous theorem we intuitively only use the linear term. In the next one we would like to also use, to some extent, the exponential term.  If we shift all the probabilities by multiplying them by a factor of $\beta > 1$,  the total success probability would be approximately $\beta \Delta^{-1/\tau}$ if $x = 1$ and $\beta (2e)^{-\beta}$ if $x =\Delta$. Thus by setting $\beta = \log_{2e}\Delta/\tau$ we maximize both these values. 
\begin{algorithm}
	\TitleOfAlgo{\algnametwo$(r,\bar{\tau})$}
	\lFor{$i\leftarrow 1$ \KwTo $\bar{\tau}$}{$p_i \gets \Delta^{-i/\bar{\tau}} \cdot \log_{2e}{\Delta}/\bar{\tau}$}
	\Repeat($r$ times){
		\procname$(\bar{\tau},p_1,p_2,\dots,p_{\bar{\tau}})$ }
\end{algorithm}

The following lemma makes this above intuition precise and gains a log-factor in performance in algorithm $\algnametwo$ (Fast Robust Local Broadcast) compared to $\algnameone$. 
As part of this analysis, we add a second statement to our lemma that will prove useful during our subsequent analysis of global broadcast. The correctness of this second lemma is a straightforward consequence of the analysis. \\
\InConference{\vspace*{-1mm}}
\begin{lemma}
	\label{lem:upper_2}
	Fix any receiver $u \in R$ and error bound $\epsilon > 0$. It follows:
	\begin{enumerate}
		\item\label{statement1} \algnametwo$(2\lceil \ln (1/\epsilon) \rceil\cdot \lceil 4 \Delta^{1/\bar{\tau}}\bar{\tau}/\log_{2e}\Delta \rceil,\bar{\tau})$ completes local broadcast with a single receiver in time \\ $O\left(\frac{\Delta^{1/\bar{\tau}} \cdot \bar{\tau}^2}{\log\Delta}\cdot  \log{(1/\epsilon)}\right)$ with probability at least $1-\epsilon$, for any $\epsilon > 0$,
	\item\label{statement2}	\algnametwo$(2,\bar{\tau})$ completes local broadcast with a single receiver with probability at least $\frac{\log_{2e}\Delta}{4 \Delta^{1/\bar{\tau}}\bar{\tau}}$.
	\end{enumerate}
\end{lemma}	
\InConference{
\begin{proof}[Proof Idea]
The proof is similar to the one of 	Lemma~\ref{lem:upper_1}. We define the probabilities $q_i$ and events $Q_i$ in the same way. The key difference is in the evaluation of the probability of success in round $t_i$ conditioned on $Q_i$ ($\Pro{S_i = 1\mid Q_i}$).  Event $Q_i$ restricts the number of neighbors connected to $u$ to some interval. We prove that the success probability $\Pro{S_i = 1\mid Q_i}$ is lower bounded by the minimum of the values at the endpoints of this interval. This is true because when $x$ stations transmit with probability $p$ to a common neighbor then the probability of a successful transmission seen as a function of $x$ has a single maximum at $x = 1/p$ hence its value at any point of some fixed interval is lower bounded by the minimum of the values at the endpoints.
\end{proof}}
\JournalProof{\ProofThree}{
\begin{proof}
		It is sufficient to prove the claim for $\tau \leq \log_{2e}\Delta$. For $\tau > \log_{2e}\Delta$ we use the algorithm for $\tau = \log_{2e}\Delta$.  Note that any algorithm that is correct for some $\tau$ must also work for any larger $\tau$ because the adversary may not choose to change the distribution as frequently as it is permitted to. In the case where $\tau \leq \log_{2e}\Delta$ we get that $\Delta^{1/\tau} \geq 2e$. Moreover $\Delta^{-1/\tau}\log_{2e}\Delta / \tau  = (2e)^{-\log_{2e}\Delta/\tau} \log_{2e}\Delta/\tau \leq 1/(2e)$ because $\log_{2e}\Delta/\tau \geq 1$ hence $p_i \leq 1/2$.
		
	 We want to show that if the nodes from $N_{u}\cap B$ execute the procedure $\procname(\tau,p_1,\dots,p_{\tau})$ twice, then $u$ receives some message with probability at least $\log\Delta/(4 \Delta^{1/\tau}\tau)$. Since we execute $\procname$ twice, we have a total of $2\tau$ consecutive time slots out of which, by the definition of our model, at least $\tau$ consecutive slots have the same distribution of the edges in $E'\setminus E$ and moreover stations try all the probabilities $p_1,p_2,\dots,p_{\tau}$.  (not necessarily in this order). Let $T$ denote the set of these $\tau$ time slots and for $i= 1,2,\dots,\tau$ let $t_i \in T$ be the step in which probability $p_i$ is used. We also denote the distribution used in steps from set $T$ by $\mathcal{E^{(T)}}$. Observe from the definition of the algorithm that during these slots the number of participating stations does not change. Hence we can denote the edges between $u$ and its neighbors that have some message by $E_{part} = \{(u,b): b\in B\}\cap E'$. We know that the edge sets are chosen independently from the same distributions: $E_t \sim \mathcal{E^{(T)}}$ for $t \in T$. Let us denote by $X_t = |E_t\cap E_{part}|$ the random variable being the number of neighbors that are connected to $u$ in step $t$ and belong to $B$. For each $i$ form $1$ to $\tau$, we define $q_i$:	
	\[
	q_i = \Pro{\Delta^{(i-1)/\tau} < X_t \leq \Delta^{i/\tau}}.
	\]
	for any $t\in T$. Observe that probabilities $q_i$ do not depend on $t$ during the considered $\tau$ rounds. Moreover $E_t\cap E_{part} \neq \emptyset$, hence $\Pro{X_t = 0} = 0$ thus:
	\begin{equation}
	\label{eqn:sum_one_2}
	\sum_{i = 1}^{\tau} q_i = 1.
	\end{equation}
	 We would like to lower bound the probability that $v$ receives a message in step $t_i$ for $i = 1,2,\dots,\tau$:
	\begin{equation}
	\label{eqn:s_i_2}
	\Pro{R_{t_i}^{(v)}}\geq  \frac{q_i \log_{2e}\Delta}{2\Delta^{1/\tau}\tau}.
	\end{equation}
	In $t_i$-th slot each station with a message transmits independently with probability is $p_{i} = \Delta^{-i/\tau}\cdot \log_{2e}\Delta/\tau$ and the transmission choices done by the stations are independent from the choice of edges $E_{t_i}$ active in round $t_i$. Let $Q_i$ denote the event that $\Delta^{(i-1)/\tau} < X_{t_i} \leq \Delta^{i/\tau}$. We have $p_i \leq 1/2$ hence we can use Lemma~\ref{lem:prosing} and~\ref{lem:interval}:
	\begin{align*}
		\Pro{R_{t_i}^{(v)}} &\geq \Pro{R_{t_i}^{(v)}\mid  Q_i}\cdot \Pro{Q_i} \\&\geq q_i \cdot \min \left\{\frac{(\Delta^{(i-1)/\tau} +1) \Delta^{-i/\tau}\log_{2e}\Delta/\tau}{(2e)^{(\Delta^{(i-1)/\tau} +1) \Delta^{-i/\tau}\log_{2e}\Delta/\tau}}, \frac{\log_{2e} \Delta /\tau}{ (2e)^{\log_{2e}\Delta/\tau}}\right\}.
	\end{align*}
	Note that $(\Delta^{(i-1)/\tau} +1) \Delta^{-i/\tau}\log_{2e}\Delta/\tau = \Delta^{-1/\tau} \cdot \log_{2e}\Delta/\tau  +  \Delta^{-i/\tau} \cdot \log_{2e}\Delta/\tau  \leq 2(2e)^{-\log_{2e}\Delta/\tau}\cdot \log_{2e}\Delta/\tau \geq 1/e$, because $\log_{2e}\Delta/\tau \geq 1$,  hence:
		\begin{align*}
	\Pro{R_{t_i}^{(v)}} &\geq q_i\min \left\{ \frac{\Delta^{-1/\tau} \log_{2e}\Delta/\tau}{(2e)^{1/e}}, \frac{\log_{2e}\Delta}{\Delta^{1/\tau}\tau} \right\} \geq \frac{q_i \log_{2e}\Delta}{2\Delta^{1/\tau}\tau}.
		\end{align*}
	
	%
	Since the edge sets are chosen independently in each step and the choices of the stations are also independent we have:
	\begin{alignat*}{3}
	\Pro{\bigwedge_{i=1}^{\tau} \left(\neg R_{t_i}^{(v)}\right)} &= \prod_{i = 1}^{\tau}\Pro{\neg R_{t_i}^{(v)}} & \qquad\text{by independence}\\
	&\leq \prod_{i = 1}^{\tau} \left(1- \frac{q_i\log_{2e}\Delta}{2 \Delta^{1/\tau}\tau}\right) &\text{by Equation~\eqref{eqn:s_i_2}}\\
	&\leq 1 - \sum_{i = 1}^{\tau}\frac{q_i \log_{2e}\Delta}{2 \Delta^{1/\tau}\tau} + \sum_{1\leq i < j \leq n}\frac{q_i q_j \log_{2e}^2\Delta}{4 \Delta^{2/\tau}\tau^2} & \text{by Lemma~\ref{lem:wpi}}\\
	&\leq 1 - \frac{\log_{2e}\Delta}{2 \Delta^{1/\tau}\tau}\sum_{i = 1}^{\tau}q_i  + \frac{\log_{2e}^2\Delta}{4 \Delta^{2/\tau}\tau^2}\left(\sum_{i=1}^n q_i\right)^2 & \\	
	&= 1 - \frac{\log_{2e}\Delta}{2 \Delta^{1/\tau}\tau} +\frac{\log_{2e}^2\Delta}{4 \Delta^{2/\tau}\tau^2}, &\text{by Equation~\eqref{eqn:sum_one_2}}\\
	&\leq 1 - \frac{\log_{2e}\Delta}{4 \Delta^{1/\tau}\tau}, 
	\end{alignat*}
\\
where the last inequality is true since if we denote $\tau = (\log_{2e} \Delta)/\alpha$ (for $\alpha\geq 1$) then we have $\Delta^{1/\tau} \tau = (2e)^{\alpha}\log \Delta/(2\alpha) \geq \log\Delta $ hence $\frac{\log_{2e}^2\Delta}{4 \Delta^{2/\tau}\tau^2} \leq \frac{\log_{2e}\Delta}{4 \Delta^{1/\tau}\tau}$.
This completes proof of~\ref{statement2}. To prove~\ref{statement1} we observe that if we execute the procedure for $ 2\tau\lceil \ln (1/\epsilon) \rceil\cdot \lceil 4 \cdot \Delta^{1/\tau}\tau/\log_{2e}\Delta \rceil$ time steps we have at least $\lceil \ln (1/\epsilon) \rceil\cdot \lceil 4 \cdot \Delta^{1/\tau}\tau/\log_{2e}\Delta \rceil$ sequences of $\tau$ consecutive time steps in which the distribution over the unreliable edges is the same and the algorithm tries all the probabilities $\{p_1,p_2,\dots,p_{\tau}\}$. Each of these procedures fails independently with probability at most $1- \log_{2e}\Delta/(4 \Delta^{1/\tau}\tau)$ hence the probability that all the procedures fail is at most:
	 \[
	 \left(1-\frac{\tau\log_{2e}\Delta}{4  \cdot \Delta^{1/\tau}\tau}\right)^{\lceil \ln (1/\epsilon) \rceil\cdot \lceil 4 \cdot \Delta^{1/\tau}\tau/\log_{2e}\Delta \rceil} \leq e^{-\lceil \ln (1/\epsilon)\rceil} < \epsilon
	 \]
\end{proof}
}
In Lemmas~\ref{lem:upper_1} and~\ref{lem:upper_2} we studied the fate of a single receiver in $R$ during an execution of algorithms $\algnameone$ and $\algnametwo$. Here we apply this result to bound the time for all nodes in $R$ to receive a message, therefore solving the local broadcast problem. In particular, 
for a desired error bound $\epsilon$,
if we apply these lemmas with error bound $\epsilon'=\epsilon/n$,
then we end up solving the single node problem with a failure probability upper bounded by $\epsilon/n$.
Applying a union bound, it follows that the probability that any node from $R$ fails to receive a message is less than $\epsilon$.
Formally:


	
\begin{theorem}
	\label{thm:mult_receivers}
	Fix an error bound $\epsilon>0$. It follows that algorithm
	$ \algnametwo(2\lceil \ln (n/\epsilon) \rceil\cdot \lceil4 \Delta^{1/\bar{\tau}}\bar{\tau}/\log\Delta \rceil)  $
	solves local broadcast in  $O\left(\frac{\Delta^{1/\bar{\tau}} \cdot \bar{\tau}^2}{\log_{2e}\Delta}\cdot  \log{(n/\epsilon)}\right)$  rounds, with probability at least $1-\epsilon$.
\end{theorem}

\InConference{\vspace*{-3mm}}
\subsection{Lower bound}
\InConference{\vspace*{-1mm}}

Observe that for $\tau = \Omega(\log\Delta)$,  \algnametwo\xspace  has a time complexity of $O(\log\Delta \log n)$  rounds for $\epsilon = 1/n$,
which matches the performance of the optimal algorithms for this problem in the standard radio model.
This emphasizes the perhaps surprising result that even large amounts of topology changes do not impede simple uniform broadcast strategies,
so long as there is independence between nearby changes.

Once $\tau$ drops below $\log{\Delta}$, however,
a significant gap opens between our model and the standard radio network model.
Here we prove that gap is fundamental for any uniform algorithm in our model.

	In  the local broadcast problem, a receiver from set $R$ can have between $1$ and $\Delta$ neighbors in set $B$. The neighbors should optimally use probabilities close to the inverse of their number. But since the number of neighbors is unknown, the algorithm has to check all the values. If we look at the logarithm of the inverse of the probabilities (call them \emph{log-estimates}) used in Lemma~\ref{lem:upper_1} we get $i \log\Delta/\tau$, for $i = 1,2,\dots,\tau$---which are spaced equidistantly on the interval $[0,\log\Delta]$.  The goal of the algorithm is to minimize the maximum gap between two adjacent log-estimates placed on this interval since this maximizes the success probability in the worst case.
With this in mind, in the proof of the following lower bound, we look at the dual problem.
Given a predetermined sequence of probabilities
used by an arbitrary uniform algorithm,
we seek the largest gap between adjacent log-estimates, and then select edge distributions that take advantage of this weakness.

\begin{theorem}
	\label{thm:lower}
	Fix a maximum degree $\Delta \geq 10$, stability factor $\tau \leq \log(\Delta-1)/16$, and uniform local broadcast algorithm $\mathcal{A}$. Assume that $\mathcal{A}$ guarantees with probability at least $1/2$ to solve local broadcast in $f(\Delta, \tau)$ rounds when executed in any dual graph network with maximum degree $\Delta$ and fading adversary with stability $\tau$. It follows that $f(\Delta, \tau) \in \Omega( \Delta^{1/\tau} \tau/\log{\Delta})$. 
	
\end{theorem}
\InConference{\vspace*{-4mm}}
\InConference{
\begin{proof}[Proof Idea]
	In this proof we use a star with $\Delta$ arms out of which only one is reliable -- all other arms are controlled by the adversary. The single receiver $u$ is the center of the star. For any uniform algorithm we divide the probabilities $p_i$ into sequences of length $\tau$ and find a distribution in which the degree of $u$ is ``hard'' for each sequence. The algorithm places $\tau$ log-estimates on interval $[0,\log\Delta]$ we, as an adversary, can clearly find a largest gap between adjacent log-estimates of length approximately $\log\Delta/\tau$. We choose the degree $d$ of $u$ such that its logarithm is inside this gap (in correct distances from both its endpoints). With this choice we can upper bound the probability of a successful transmission in any step during these $\tau$ steps, because the distance between the log-estimate and the logarithm of the degree of $u$ gives us lower bound on $dp_i$ if $p_i > 1/d$ or of $1/(dp_i)$ if $p_i < 1/d$ which in turn upper bounds the probability of a successful transmission.
\end{proof}}
\JournalProof{\ProofFive}{
\begin{proof}
	 Consider the dual graph $G = (V,E)$ and $G' = (V,E')$, defined as follows: $V = \{v,u,v_1,\dots,v_{n-2}\}$  and $E = \{(u,v_i), i\in \{1,2,\dots,\Delta-1\}\} \cup \{(v_1,v), (v,v_{\Delta})\} \cup \{(v_i,v_{i+1}), i \in \{\Delta,\dots,n-3\}\}$ and $E' = E\cup\{(v_i,v), i\in\{2,3,\dots,\Delta-1\}\}$ (see Figure~\ref{fig:lower}). We will study local boadcast in this dual graph with $B  = \{u,v_1,v_2,\dots,v_{\Delta-1}\}$ and $R = \{v\}$.
	\begin{figure}
	\centering
	\includegraphics[width=0.6\linewidth]{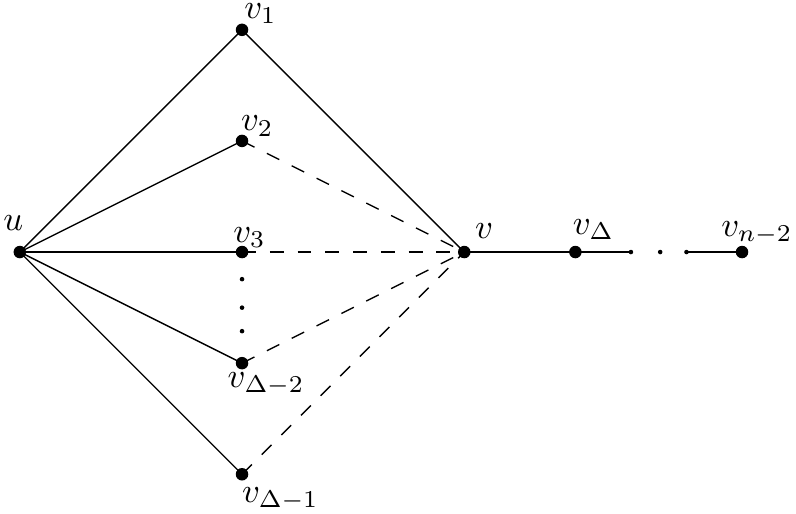}
	\caption{A graph used in proofs of Theorems~\ref{thm:lower} and~\ref{thm:lower_2}. Solid lines correspond to edges in $E$ and dashed lines correspond to edges in $E'\setminus E$ (unreliable edges).}
	\label{fig:lower}
\end{figure}
Observe that the maximum degree of any node is indeed $\Delta$ and the number of nodes is $n$. Nodes $v_{\Delta},v_{\Delta+1},\dots,v_{n-2}$ do not belong to $B\cup R$ hence are not relevant in our analysis. 

	Using the sequence of probabilities $p_1,p_2,\dots$ used by algorithm $\mathcal{A}$ we will define a sequence of distributions over the edges that will cause a long delay until node $v$ will receive a message. The adversary we define is allowed to change the distribution every $\tau$ steps. Accordingly, we partition the rounds into {\em phases} of length $\tau$, which we label $1,2,3,\dots$. Phase $k$ consists of time steps $I_k =\{1+k\cdot \tau,2+k\cdot\tau,\dots,(k+1)\cdot\tau \}$. For each phase $k \geq 1$, the adversary will use a distribution $\mathcal{D}_k$ that's defined with respect to the probabilities used by $\mathcal{A}$ during the rounds in phase $k$. In particular, let sequence $P_k = \{p_{i}\}_ {i\in I_k}$ be the $\tau$ probabilities used by $\mathcal{A}$ during phase $k$. 
	
	We use $P_k$ to define the distribution $\mathcal{D}_k$ as follows. Define $\DDelta = \Delta - 1$ and let $N$ represent $\lfloor\log{\DDelta}\rfloor$ urns labeled with numbers from $1$ to $\lfloor\log{\DDelta}\rfloor$. Into these urns we place balls with numbers $\lceil\log(1/p_j)\rceil$ and $\lfloor\log(1/p_j)\rfloor$ for all $ j\in I_k$. Ball with number $i$ is placed into the bin with the same number. With this procedure for each $j$, we place two balls in adjacent bins if $\lceil\log(1/p_j)\rceil \neq \lfloor\log(1/p_j)\rfloor$ and a single ball in the opposite case. We arrange the bins in a circular fashion i.e., bins $\lfloor \log{\DDelta} \rfloor$ and $1$ are consecutive and we want to find the longest sequence of consecutive empty bins. Observe that since for each $j$ we put either a single ball or two balls into adjacent bins we have at most $\tau$ sequences of consecutive empty bins. Moreover, since at most $2\tau$ bins contain a ball then there exists a sequence of consecutive empty bins of length at least $\frac{\lfloor\log\DDelta\rfloor - 2\tau}{\tau}$. Knowing that $\tau$ is an integer and that $\tau \leq \log\DDelta / 16$ we can represent $\log\DDelta = a \tau + b + \{\log\DDelta\}$, where $\{\log\DDelta\}$ is the fractional part of $\log\DDelta$ and $b + \{\log\DDelta\} < \tau$. We can then show that:
	 \[
	 \frac{\lfloor\log\DDelta\rfloor - 2\tau}{\tau}  = a + \frac{b}{\tau} - 2 \geq  \left\lfloor \frac{\log\DDelta}{\tau} \right\rfloor - 2.
	  \]
    We define:
	\begin{align*}
	x &= \lfloor \log\DDelta /\tau\rfloor - 3 - \lfloor\log(\lfloor \ln\DDelta/\tau\rfloor)\rfloor, \\
	y &= \lfloor\log{(\lfloor \ln\DDelta/\tau\rfloor)}\rfloor + 1.
	\end{align*}
	We observe that for $\tau \leq \log{\DDelta}/16$ we have $\log{(\lfloor \ln\DDelta/\tau\rfloor)} \geq 4$ hence $x$ and $y$ are both positive integers and moreover $x+y = \lfloor \log\DDelta /\tau \rfloor - 2$. Hence we already showed that there exists a sequence of consecutive empty bins of length at least $x+y$. Now, we pick the label of $(y+1)$-st bin in this sequence (the order of bins is according to the circular arrangement i.e., $1$ comes after $\lfloor\log\DDelta \rfloor$) and call it $a_k$. Let $A_k = \{\log(1/p_j) : j\in I_k\}$. This set contains logarithms of all the estimates ``tried" by the algorithm in $k$-th phase.  Now we split $A_k$ into elements that are larger and that are smaller than $a_k$: $A_k = A^{(\geq)}_{k} \cup A^{(<)}_{k}$, $A^{(\geq)}_{k} = \{a \in A_k : a \geq a_k\}$, $A^{(<)}_k = \{a \in A_k : a < a_k\}$. We observe that if $a \in A^{(<)}_k$ then $a \leq a_k - y$ because there are $y$ empty bins between bin $a_k$ and the bin containing ball $\lceil a \rceil$. Symmetrically if $a \in A^{(\geq)}_k$ then $a \geq a_k + x-1$ because there are $x-1$ empty bins between bin $a_k$ and the bin containing ball $\lfloor a \rfloor$. 
	
	
	 In our distribution $\mathcal{D}_k$ in phase $k$, we include all edges from $E$, plus a subset of size $2^{a_k} -1$ selected uniformly from $E' \setminus E$. This is possible since the adversary can choose to activate any subset of links among the set $\{(v_i,v), i\in\{2,\dots,\DDelta\}\}$. With this choice, the degree of $v$ is $2^{a_k}$ in phase $k$ hence we can bound the probability that a successful transmission occurs in phase $k$.
	 
	 Having chosen the distribution of the edges between $v$ and $\{v_1,v_2,\dots,v_{\DDelta}\}$ we can now bound the probability of a successful transmission in any step $t$ in the considered phase. Let the event of a successful transmission in step $t$ be denoted by $S_t$. For this event to happen exactly one of the $2^{a_k}$ nodes among $\{v_1,v_2,\dots,v_{\DDelta}\}$ that are connected to $v$ need to transmit. We have:
	  \[\Pro{S_t} = 2^{a_k} p_t \cdot (1-p_t)^{2^{a_k}-1}.\]
	 Take any step $t$ and the corresponding probability $p_t$ used by the algorithm. We know that $a_k$ is chosen so that $a_k \geq \log(1/p_t) + y$ or $a_k \leq \log(1/p_t) - x$. We consider these cases separately:

	 \begin{description}
	 	\item [Case 1: \boldm{$a_k \leq \log(1/p_t) - x + 1$}]
	   \begin{align*}
	 	\Pro{S_t} &= 2^{a_k}\cdot p_t\cdot\left(1-p_t\right)^{2^{a_k}-1} \leq 2^{a_k - \log(1/p_t)} \leq 2^{-x+ 1} \\& = 2^{-\lfloor \log\DDelta /\tau\rfloor + 4 + \lfloor\log(\lfloor \ln\DDelta/\tau\rfloor)\rfloor}  \leq 2^{-\log\DDelta /\tau + \log(\lfloor \ln\DDelta/\tau\rfloor) + 5 }
	 	\\&\leq \frac{32 \lfloor \ln\DDelta/\tau\rfloor}{\DDelta^{1/\tau}}  \\&\leq \frac{32\ln{\DDelta}}{\DDelta^{1/\tau}\tau}.
	 	\end{align*}
	    \item [Case 2: \boldm{$a_k \geq \log(1/p_t) + y$}]
	    	 \begin{align*}
	    \Pro{S_t} &= 2^{a_k}\cdot p_t\cdot\left(1-p_t\right)^{2^{a_k}-1} =  \frac{2^{a_k}p_t}{1 - p_t } (1-p_t)^{2^{a_k}} \leq \frac{2^{a_k}p_t}{1 - p_t } e^{-2^{a_k}p_t}.
	    \end{align*}
	    We know that $a_k\geq \log(1/p_t) + y$ and $a_k \leq \log\DDelta$ thus $p_t \leq 2^{y}/\DDelta$ hence since $\DDelta \geq 9$ we get $1/(1-p_t)\leq 1/2$. Moreover since $2^{a_i} p_t \geq 2^y \geq 4$ we have $e^{-2^{a_i}p_t/2} < 1/(2^{a_i}p_t) $ (because $e^{x/2} > x$ for all $x$). Which gives in this case:
	    \begin{align*}
	    \Pro{S_t} & <  2 e^{-2^{a_k}p_t/2} \leq e^{-2^{y-1}} \leq 2e^{-\lfloor \ln\DDelta/\tau\rfloor} \leq
	    2\DDelta^{-1/\tau} \leq \frac{32 \ln\DDelta}{\DDelta^{1/\tau}\tau}.
	    \end{align*}
	 \end{description}

	 We have just shown that the probability that $v$ receives a message in our fixed phase $k$ is at most $\frac{32 \ln\DDelta}{\DDelta^{1/\tau}\tau}$. To conclude the proof, we apply a union bound to show that probability $v$ receives a message in at least one of $\DDelta^{1/\tau} \tau /(64 \ln \DDelta) - 1$ phases, which require $\DDelta^{1/\tau} \tau^2 /(64 \ln \DDelta) - \tau$ total rounds, is strictly less than 1/2:
	  \[
	 \Pro{\bigcup_{t=1}^{\DDelta^{1/\tau} \tau /(64 \ln \DDelta) - 1} S_t} \leq \sum_{t = 1}^{\DDelta^{1/\tau} \tau/(64 \ln \DDelta) - 1} \Pro{S_t} < \frac{1}{2}.
	 \]
	 
\end{proof}}
In our next theorem, we refine the argument used in Theorem~\ref{thm:lower} for the case where $\tau$ is a non-trivial amount smaller than the $\log{\Delta}$ threshold. 
We will argue that for smaller $\tau$, the complexity is $\Omega(\Delta^{1/\tau} \tau^2 /\log{\Delta})$, which more exactly matches our best upper bound. 
We are able to trade this small amount of extra wiggle room in $\tau$ for a stronger lower bound because it simplifies certain probabilistic obstacles in our argument.
Combined with our previous theorem,
the below result shows our upper bound performance is asymptotically optimal for uniform algorithms for all but a narrow range of stability factors,
for which it is near tight.


 %
\begin{theorem}
	\label{thm:lower_2}
	Fix a maximum degree $\Delta \geq 10$, stability factor $\tau \leq \ln(\Delta-1)/(12\log\log (\Delta-1))$, and uniform local broadcast algorithm $\mathcal{A}$. Assume that $\mathcal{A}$ guarantees with probability at least $1/2$ to solve local broadcast in $f(\Delta, \tau)$ rounds when executed in any dual graph network with maximum degree $\Delta$ and fading adversary with stability $\tau$. It follows that $f(\Delta, \tau) \in \Omega(\Delta^{1/\tau} \tau^2 /\log{\Delta})$.	
\end{theorem}
\InConference{
\begin{proof}[Proof Idea]
	The proof is similar to proof of Theorem~\ref{thm:lower}. Here we also find a gap of length $\log\Delta/\tau$ and then we argue that in a ``proximity'' of each such a large gap there has to exist a large number of log-estimates. The proximity is defined so that all log-estimates outside of it are (almost) irrelevant, give a very small probability of success, if we choose the logarithm of the degree of $u$ to be inside the considered gap. This in turn implies that in the remaining part of the interval the ``density" of log-estimates is lower hence there must exist another large gap. By repeating this argument we can derive a contradiction with the assumed time complexity. 
	The reason why we need to restrict $\tau$ is that our defined proximity must be of the same order as $\log\Delta/\tau$ which is no longer true for $\tau$ being close to $\log\Delta$.
\end{proof}}
\JournalProof{\ProofSix}{
\begin{proof}
	In this proof we will use the same graph as in Theorem~\ref{thm:lower}. Let $G = (V,E)$ and $G' = (V,E')$. Let $V = \{v,u,v_1,\dots,v_{n-2}\}$ and let $R = \{u,v_1,v_2,\dots,v_{\Delta-1}\}$ and $E = \{(u,v_i), i\in \{1,2,\dots,\Delta\}\} \cup \{(v_1,v), (v,v_{\Delta})\} \cup \{(v_i,v_{i+1}), i \in \{\Delta,\dots,n-3\}\}$ and $E' = E\cup\{(v_i,v), i\in\{2,3,\dots,\Delta\}\}$ (see Figure~\ref{fig:lower}).

	Let $p_1,p_2,\dots$ be the fixed sequence of broadcast probabilities used by nodes in $B$ running $\mathcal{A}$. Using this sequence we will define a sequence of distributions over the edges that will cause a long time for this algorithm until node $v$ will receive a message.  
	
	The adversary is allowed to change the distribution once every $\tau$ steps. Therefore we will define the $k$-th distribution $\mathcal{D}_k$ based on sequence $P_k = (p_{(k-1)\tau  +  1}, p_{(k-1)\tau  +  2},\dots, p_{k\tau })$ of probabilities and distribution $\mathcal{D}_k$ will be used in rounds $(k-1)\tau  +  1, (k-1)\tau  +  2,\dots, k\tau$.  Consider intervals of $\tau$ time steps (call such interval a \textit{phase}) and the corresponding probabilities $p_{j + i\cdot\tau}$ ($j<\tau$). Let us fix any phase $k$ and consider values $l_i = \log(1/p_{i+ (k-1)\tau})$, for $i=1,2,\dots,\tau$. We denote $\DDelta = \Delta - 1$. As an adversary we are allowed to define an integer value $l^* \in [1,2,\dots,\lfloor\log \DDelta \rfloor]$ based on the $l$-values and define a distribution for phase $k$ in which there are always $2^{l^*}$ active links between nodes $v_1,v_2,\dots,v_{n}$ and $v$. The success probability in $i$-th step of the considered phase is then
\[
s_i = p_{i + (k-1)\tau} \cdot 2^{l^*} \cdot (1-p_{i + (k-1)\tau})^{2^{l^*}-1}.
\]
Our goal as an adversary is to find such $l^*$ that minimizes $\sum_{i=1}^{\tau} s_i$. We will show that it is always possible to find such $l^*$ that $\sum_{i=1}^{\tau} s_i = O(\DDelta^{-1/\tau} \log \DDelta/\tau) = \Theta(\Delta^{-1/\tau} \log \Delta/\tau)$. This will give us that $\Omega(\DDelta^{1/\tau}\tau/ \log \DDelta)$ phases of $\tau$ steps hence in total $\Omega(\Delta^{1/\tau}\tau^2/ \log \Delta)$ steps are needed to complete local broadcast with constant probability. 

Assume by contradiction that there exists a choice of $l_1,l_2,\dots,l_{\tau}$ such that for any choice of $l^*$ we have that $\sum_{i=1}^{\tau} s_i \geq c\frac{\log \DDelta}{\DDelta^{1/\tau}\tau}$, where $c = 2409$. We fix this choice of $l$-values $l_1,l_2,\dots,l_{\tau}$ and we denote:
\begin{align*}
x &= \log \DDelta / \tau + \log \tau - \log\ln \DDelta, \\
y &= \log(\ln \DDelta/\tau), \\
x' &= \log \DDelta/\tau + 2\log\tau -\log\ln \DDelta, \\
y' &= \log(\ln \DDelta/\tau + 2\ln \tau).
\end{align*}
Since $\tau < \ln \DDelta/(12\log\log \DDelta)$ we have that $y \geq 3$. Observe that $x + y = \log \DDelta/\tau$ and $x' + y' \leq \log \DDelta/\tau + 2\log\tau -\log\ln \DDelta+ \log(\ln \DDelta/\tau + 2 \ln \tau) = \log \DDelta/\tau + 2\log\tau + \log(1/\tau + 2\ln\tau/\ln \DDelta) \leq \log \DDelta/\tau + 2\log\tau + \log 3$. Let $\Delta^*$ denote the number of active links between $v$ and $v_1,v_2,\dots,v_{\DDelta}$ in the considered phase and $l^* = \log \Delta^*$. In any step if $\Delta^*$ is such that $l^* \geq y'$ then if we also have $p_i \geq \frac23$ then we get:
\begin{equation}
\label{eqn:bigProb}
s_i  = p_i 2^{l^*}(1-p_i)^{2^{l^*}-1 }\leq \frac{2^{l^*}}{3^{2^{l^*}- 1}} \leq \frac{3(\ln\DDelta/\tau + 2\ln\tau) }{e^{\ln\DDelta/\tau + 2\ln\tau}} \leq \frac{9\ln \DDelta}{\DDelta^{1/\tau}\tau^2}, 
\end{equation}
where the last inequality is true because $\tau \geq 1$ and $\ln \tau \leq \ln\DDelta$. This shows that the sum of all such $s_i$ is at most $\frac{9\ln \DDelta}{\DDelta^{1/\tau}\tau}$. Consider now only steps with $p_i < 2/3$. Then:
\[
s_i =  2^{l^* - l_i} \left(1-\frac{1}{2^{l_i}}\right)^{2^{l^*} - 1}  = 2^{l^* - l_i} \cdot \frac{\left(1 - \frac{1}{2^{l_i}}\right)^{2^{l^*}}}{1-p_i} \leq 3 \cdot 2^{l^* - l_i} \cdot e^{-2^{l^* - l_i}}
\]
\begin{align}
\label{eqn:x} &\Pro{\Single \mid l_i \geq l^* + x - 3} \leq 3 \cdot 2^{-x + 3}\leq \frac{24 \ln \DDelta}{\DDelta^{1/\tau}\tau} \\
\label{eqn:xprime} &\Pro{\Single \mid  l_i \geq l^* + x'-1}     \leq 3 \cdot 2^{-x'+1} \leq  \frac{6 \ln \DDelta}{\DDelta^{1/\tau}\tau^2} \\
\label{eqn:y} & \Pro{\Single \mid  l_i \leq l^* - y } \leq \frac{3 \cdot 2^{y}}{e^{2^y}} = \frac{3\ln \DDelta}{\DDelta^{1/\tau}\tau} \\
\label{eqn:yprime} & \Pro{\Single \mid  l_i \leq l^* - y'} \leq \frac{3 \cdot 2^{y'}}{e^{2^{y'}}}  = \frac{3 (\ln \DDelta/\tau + 2\log\tau)}{\DDelta^{1/\tau} \tau^2} \leq \frac{9 \ln \DDelta}{\DDelta^{1/\tau}\tau^2} 
\end{align}
Observe that for a fixed value of $l^*$, for any $i$ such that $l_i \notin [l^* - y', l^* + x'-1]$  we have $s_i \leq \frac{9\ln \DDelta}{\DDelta^{1/\tau} \tau^2}$ (by Equations~\eqref{eqn:xprime},~\eqref{eqn:yprime}). Hence the sum of all such values $s_i$ is at most $\frac{9\ln \DDelta}{\DDelta^{1/\tau} \tau}$. Hence we only need to find such $l^*$ that the sum of the values $s_i$ for which the corresponding $l_i \in [l^*_1 - y', l^*_1 + x']$ is less than $\frac{(c-9)\ln \DDelta}{\DDelta^{1/\tau} \tau}$.

 We denote the smallest and the largest $l$-values: $l_{sm} = \min_{i\in \{1,2,\dots,\tau\}}\{l_i\}$ and $l_{lg} = \max_{i \in \{1,2,\dots,\tau\}} \{l_i\}$. We will prove two following claims about $l_{sm}$ and $l_{lg}$:

\begin{description}
	\item[\boldm{$l_{sm} \leq x'$}] Observe that otherwise we can choose $l^* = 0$ ($\Delta^*$ is then equal to $1$ which corresponds to exactly one active link between $\{v_1,v_2,\dots,v_{\DDelta}\}$ and $v$) and then by Equation~\eqref{eqn:xprime} under this choice of $l^*$ all values $s_i$ would satisfy $s_i\leq \frac{6 \DDelta^{-1/\tau} \ln \DDelta}{\tau^2}$ (because if $l_i \geq x'$ then $p_i < 2/3$). 
	\item[\boldm{$l_{lg} \geq \log \DDelta - y'$}] If it is not the case, we choose $l^* = \log\DDelta$ and by Equation~\eqref{eqn:yprime} we have that if $p_i < 2/3$ then $s_i \leq \frac{9 \ln \DDelta}{\DDelta^{1/\tau}\tau^2}$ and by Equation~\eqref{eqn:bigProb} that if $p_i\geq 2/3$ then $s_i \leq \frac{9\ln \DDelta}{\DDelta^{1/\tau}\tau^2}$. And the sum of all values of $s_i$ is at most $\frac{9 \ln \DDelta}{\DDelta^{1/\tau}\tau}$ which contradicts our assumption. 
\end{description}
Consider now interval $\Gamma_1 = [l_{sm},l_{lg}]$. Two previous claims showed that $|\Gamma_1| \geq \log\DDelta - x' - y'$.  We can now consider the placement of values $l_i$ on $\Gamma_1$ and analyze gaps between the adjacent values. Gap $g_i$ is the difference between the $(i+1)$-st smallest and $i$-th smallest value out of all values $l_j$ that belong to $\Gamma_1$. We want to show the following:
\begin{equation}
\label{eqn:maxgap}
 \max_{i} g_i \leq x' + y'
\end{equation}
 Assume on the contrary that such a gap between $l_i$ and $l_j$ exists. Then we pick $l^* = \lceil l_i + y' \rceil$ and observe that $l^*$ is an integer and is at least $y'$ larger than each smaller $l$-value and at least $x'-1$ smaller than each larger $l$-value. In such a case $l^*\geq y'$ hence for all $i$ such that $p_i \geq 2/3$ by Equation~\eqref{eqn:bigProb} we have $s_i \leq \frac{9\log \DDelta}{\DDelta^{1/\tau}\tau^2}$ and if $p_i < 2/3$ then (since $l^* \geq y'$) by Equations~\eqref{eqn:xprime}~\eqref{eqn:yprime} we have $s_i \leq \frac{9 \ln \DDelta}{\DDelta^{1/\tau}\tau^2}$. Thus if any gap has size at least $x'+y'$ then $\sum_{i=0}^{\tau}s_i \leq \frac{9 \ln \DDelta}{\DDelta^{1/\tau}\tau}$ which contradicts our assumption.

 We know that there are at most $\tau-1$ gaps and that they cover area of at least $\log \DDelta - x' - y'$. Hence we can lower bound the average length of a gap:
\begin{align*}
d_1 &= \frac{\log \DDelta - x' - y'}{\tau - 1} \geq \frac{\log \DDelta - \log \DDelta/\tau - 2\log \tau - \log 3 }{\tau -1}  \\&=  \frac{\log \DDelta (1 - 1/\tau)}{\tau (1-1/\tau)} - \frac{2\log \tau + \log 3}{\tau} \geq \frac{\log \DDelta}{\tau} - 2.
\end{align*}
Thus there exists a gap $G_1$ of length at least $d_1$. Knowing that $y\geq 3$ we have $d_1 \ge x+y -2\geq1$ and inside this gap we can find an integer value $l^*_1$ that is at least $y$ larger than the closest smaller $l$-value and at least $x-3$ smaller than the closest smaller $l$-value. Consider values of $s_i$ with this choice of $l^*$. By Equations~\eqref{eqn:x}~\eqref{eqn:y} if $l^* = l^*_1$, each $s_i$ is at most $\frac{24 \DDelta^{- 1/\tau} \ln \DDelta}{\tau}$. Consider now interval $I_1 = [l^*_1- y',l^*_1 + x']$. By Equations~\eqref{eqn:xprime} and~\eqref{eqn:yprime} for all $i$ such that $l_i \notin I_1$ and $p_i < 2/3$ we have $s_i \leq \frac{9\DDelta^{-1/\tau} \ln \DDelta}{\tau^2}$. If $p_i \geq 2/3$ then also $s_i\leq \frac{9\DDelta^{-1/\tau} \ln \DDelta}{\tau^2}$ because $l^*_1 \geq y'$. Thus the sum of all $s_i$ for which $l_i \notin I_1$ or $p_i \geq 2/3$ is at most $\frac{9\DDelta^{-1/\tau} \ln \DDelta}{\tau}$.  Since by the assumption, the sum of all $s_i$ is at least $\frac{c\ln \DDelta}{\DDelta^{1/\tau} \tau}$ then the sum of all $s_i$ for which $l_i \in I_1$ and $p_i<2/3$ has to be at least $\frac{2400 \ln \DDelta}{\DDelta^{1/\tau}\tau}$. By the choice of $l_1^*$, each $s_i$ for which $l_i\in I_1$ and $p_i \geq 2/3$ is at most $\frac{24 \ln \DDelta}{\DDelta^{1/\tau}\tau}$ hence we must have at least $100$ such $l$-values. We have shown that there are at least $100$ $l$-values inside interval $I_1$. 

We find the smallest and the largest $l$-values inside $I_1$.
\[
l^{(1)}_{sm} = \min_{i\in\{1,2,\dots,\tau\}}\{l_i : l_i\in I_1\}
\]\[
l^{(1)}_{lg}  = \max_{i\in\{1,2,\dots,\tau\}}\{l_i : l_i\in I_1\}
\]
 We consider interval $\Gamma_2 = \Gamma_1 \setminus (l^{(1)}_{sm},l^{(1)}_{lg})$ (we remove the interior of the interval $[l^{(1)}_{sm},l^{(1)}_{lg}]$ keeping the endpoints). We know that we removed at least $98$ $l$-values. Since the $l$-values have to work for any $l^*$ we can now argue about the average length of a gap inside $\Gamma_2$ and locate a different value $l^*_2$ in the remaining interval and identify $100$ $l$-values close to $l^*_2$. But we need to make sure that $|l^*_1 - l^*_2| \geq x' + y'$ since otherwise we would count the same $l$-values twice.  
 
We extend the interval $I_1$ to $I'_1 = [l^*_1 - (x' + y'),l^*_1 + (x' + y')]$ and we find the smallest $l$-value larger than any $l$-value inside $I'_1$  (call it $l'^{(1)}_{sm}$) and the largest $l$-value smaller than any $l$-value inside $I'_1$ (call it $l'^{(1)}_{lg}$). If both these values exist, we consider interval $\Gamma^*_2 = \Gamma_1 \setminus (l'^{(1)}_{sm},l'^{(1)}_{lg})$ (we remove the interior of the interval $[l'_{sm},l'_{lg}]$ keeping the endpoints). If $l'_{sm}$ does not exist (there is no such $l$-value), we define $\Gamma^*_{2} =\Gamma_1 \setminus [l_{sm},l'^{(1)}_{lg})$ and symmetrically if $l'^{(1)}_{lg}$ does not exist we set $\Gamma^*_{2} = \Gamma_1 \setminus (l'^{(1)}_{sm},l_{lg}]$. 

Now we want to show that $\Gamma^*_{2} \geq |\Gamma_1| - 5(x' + y')$. It is because $|I'_{1}| = 2(x' + y')$ and by Equation~\eqref{eqn:maxgap}, length of any gap is at most $x' + y'$ hence distance between $l^*_1 - (x' + y')$ and $l'_{sm}$ is at most $x' + y'$ (similarly between $l^*_1 + (x' + y')$ and $l'_{lg}$). If $l'_{sm}$ or $l'_{lg}$ does not exist we remove additionally no more than $x' + y'$ because the smallest $l$-value that is most $x'$ and the largest is at least $\log \DDelta - y'$. This shows that we remove the total area of at most $5(x' + y')$. 

Now we consider the average length of a gap in $\Gamma^*_{2}$. We removed at least $98$ $l$-values (because $I^*_{1}$ contains $I'_1$) and area of at most $5(x' + y') \leq 5\log \DDelta/\tau + 10 \log\tau + 5\log 3$. Hence the average length of a gap in $\Gamma^*_2$ is:
\[
d_2 \geq \frac{\log \DDelta  - 6 (x' + y')}{\tau  -  98}\geq \frac{\log \DDelta}{\tau}.
\]
We pick a gap of length at least $\frac{\log \DDelta}{\tau}$ and find an integer $l^*_2$ at least $y$ larger than the closest smaller $l$-value and at least $x-1$ smaller than the closest larger $l$-value. Observe moreover that $| l^*_2 - l^*_1| \geq x' + y'$ because $l^*_2$ does not belong to the interior of interval $I'_1$. We define $I_2 = [l^*_2 - y',l^*_2 + x']$. Observe that $I_1$ and $I_2$ are disjoint (except possibly their endpoints). We can now argue that $I_2$ also contains $100$ $l$-values similarly as $I_1$. And moreover at most one $l$-value can be shared between $I_1$ and $I_2$ (because the interiors of the intervals are disjoint). And now we extend $I_2$ to $I^*_2$, construct $\Gamma_3$ and $\Gamma_3^*$ and repeat the whole procedure. This procedure identifies at least $98$ unique $l$-values in each step hence it can last for at most $\tau/98$ iterations. But we remove at most $5\log \DDelta/\tau + 10 \log\tau + 5\log 3$ area per iteration. Since we assumed $\tau\leq \ln \DDelta/\log \log \DDelta$, then $5 \log \DDelta/\tau + 10 \log\tau + 5\log 3 \leq 10\log \DDelta/\tau$. This leads to contradiction since there are only $\tau$ $l$-values. Hence for any choice of $l_1,\dots,l_{\tau}$ there exists $l^*$ such that $\sum_{i=1}^{\tau} s_i < c \frac{\ln \DDelta}{\DDelta^{1/\tau}\tau}$. Thus by the union bound the algorithm needs to run for at least $\DDelta^{1/\tau} \tau/(c\ln \DDelta)$ phases to accumulate the total probability of success of $1/2$. Knowing that each phase lasts for $\tau$ rounds the total number of steps needed is $\Omega(\Delta^{1/\tau} \tau^2/\ln \Delta)$
\end{proof}}

\InConference{\vspace*{-2mm}}
\section{Global Broadcast}
We now turn our attention to the global broadcast problem.
Our upper bound will use the same broadcast probability sequence as our best local broadcast algorithm from before. 
As with local broadcast, for $\tau \geq \log{\Delta}$, our performance nearly matches the optimal performance in the standard radio network model,
and then degrades as $\tau$ shrinks toward $1$.
Our lower bound will establish that this degredation is near optimal for uniform algorithms in this setting.
 In this section we also use the notation $\bar{\tau} =  \min\{\tau,\lceil\log\Delta\rceil\}$.
\InConference{\vspace*{-2mm}}
\subsection{Upper Bound}
A uniform {global} broadcast algorithm requires each node to cycle through a predetermined sequence of broadcast probabilities once it becomes {\em active} (i.e., has received the broadcast message). 
The only slight twist in our algorithm's presentation is that we assume that once a node becomes active,
it waits until the start of the next probability cycle to start broadcasting.
To implement this logic in pseudocode, we use the variable  $Time$ to indicate the current global round count.
We detail this algorithm below (notice, the $\algnametwo(2)$ is the local broadcast algorithm analyzed in Lemma~\ref{lem:upper_2}). 
\InConference{\vspace*{-1mm}}
\begin{algorithm}[h]
	\TitleOfAlgo{\globalalgname$(\epsilon)$} 
	Wait until receiving the message\\
	Wait until $(Time \text{ mod }2\bar{\tau}) = 0$\\
	\Repeat($\lceil \ln{(2n/\epsilon)}\rceil \cdot \lceil 4 \Delta^{1/\bar{\tau}}\bar{\tau}/\log\Delta \rceil$  times){
			\algnametwo$(2)$ }
\end{algorithm}
\begin{theorem}
\label{thm:global_upper}	
		Fix an error bound $\epsilon > 0$. It follows that algorithm
	$ \globalalgname(\epsilon)$
	completes global broadcast in time $O\left((D+ \log (n/\epsilon))\cdot  \frac{\Delta^{1/\bar{\tau}} \bar{\tau}^2}{\log\Delta}\right)$, with probability at least $1-\epsilon$.
\end{theorem}
\InConference{
	\begin{proof}[Proof Idea]
		Here we use the same idea as in the proof of~\cite[Theorem 4]{Bar-YehudaGI92}. There a local broadcast algorithm (\Decay) is used as a black box in a global broadcast algorithm. We use a different local broadcast algorithm (\algnametwo) but the same analysis applies.
	\end{proof}
}
\JournalProof{\ProofSeven}{
\begin{proof}
	Similarly to the analysis of the local broadcast algorithms, we consider only the case of $\tau \leq \log\Delta$ since for any larger $\tau$ we use the algorithm for $\tau = \log\Delta$.
	Take any station $u$ and assume that some positive number of neighbors of $u$ in $E$ execute in parallel procedure $\algnametwo(2)$. Then by Lemma~\ref{lem:upper_2} station $u$ receives a message from some neighbor with probability at least $\frac{\ln\Delta}{4 \Delta^{1/\tau}\tau}$. Note that the same number of neighbors of $u$ have to execute both procedures \procname of $\algnametwo(2)$ and at least one of these neighbors has to be connected to $u$ by a reliable link. This is true since after receiving the message, a station waits until a time slot that is a multiple of $2\tau$ (Line $3$ in the pseudocode). Hence we can treat each execution of $\algnametwo(2)$ as a single \emph{phase}. 
	
	Let $o$ denote the originator of the message. Fix any tree $\mathcal{T}$ of shortest paths on graph $G$ (e.g., BFS Tree) on edges from $E$ (reliable) rooted at $o$. We would like to bound the progress of the message on tree $\mathcal{T}$. For any station $u$ we can denote by $p(u)$ the parent of $u$ in tree $T$. For any station $u$ we can define the earliest time step $T(u)$ in which $p(u)$ receives the message. We set $T(u) = \infty$ if the message does not reach $p(u)$. If $T(u) < \infty$ we consider $\lceil \ln(2n/\epsilon) \rceil\cdot \lceil4 \Delta^{1/\tau}\tau/\log\Delta \rceil$ phases that follow step $T(u)$.  A phase is called \emph{successful} if it succeeds in delivering the message to $u$ and \emph{unsuccessful} otherwise. Note that (assuming that $T(u) < \infty$) the probability that all phases are unsuccessful for any fixed $u$ is at most
	\begin{equation}
	\label{eqn:success}
	\left(1- \frac{\ln\Delta}{4 \Delta^{1/\tau}\tau}\right)^{\lceil \ln(2n/\epsilon) \rceil\cdot \lceil 4 \Delta^{1/\tau}\tau/\log\Delta \rceil} \leq e^{-\lceil \ln(2n/\epsilon) \rceil} = \frac{\epsilon}{2n}.
	\end{equation}
	Let us denote by $S$ an event that $T(u)  < \infty$ for all stations $u$. And by $S_i$ the event that $T(u) < \infty$ for all stations at distance at most $i$ from the root in tree $\mathcal{T}$. If $d_i$ denotes the number of stations at distance $i$ from the root in tree $\mathcal{T}$ we get:
	\begin{alignat*}{3}
	\Pro{S}  &= \Pro{S_D} \geq \Pro{S_{D}|S_{D-1}} \Pro{S_{D-1}} \\
	             &\geq \Pro{S_1} \prod_{i=2}^{D}\Pro{S_{i}|S_{i-1}} \\
	             & \geq \prod_{i=1}^D\left(1 - \frac{\epsilon d_i}{2n} \right) &\text{by~\eqref{eqn:success} and union bound}\\
	             & \geq 1- \sum_{i=1}^D \frac{\epsilon d_i}{2n} &\text{by Lemma~\ref{lem:wpi}} \\
	             & = 1 - \frac{\epsilon}{2}.
	\end{alignat*}
	If event $S$ takes place, the message reaches all the nodes of the network. Clearly it can reach node $u$ not necessarily from its parent $p(u)$ in tree $\mathcal{T}$, but this would only help in our analysis (it will cause the message to arrive at $u$ faster). Now we want to bound the number of phases it takes for the message to traverse a path in the tree.  Fix any station $u$ and let $\mathcal{P} = (o, v_1,v_2,\dots,v_{D'-1}, u)$ denote the path from $o$ to $u$ in tree $T$ (note that $D' \leq D$). We denote by $R_i$ the round in which $v_i$ receives the message ($R_{D'}$ denotes the round in which $u$ receives the message) and introduce random variables $\Delta_i = \max\{0, R_{i} - R_{i-1}\}$. Conditioning on event $S$, variables $\Delta_i$ are stochastically dominated by independent geometric random variables with success probability $\frac{\ln\Delta}{4 \Delta^{1/\tau}\tau}$. We have $D'$ such variables and the probability that sum $T$ of them exceeds $L = 4(D' + \ln(2n/\epsilon))\cdot  \frac{7 \Delta^{1/\tau} \tau}{\log\Delta} = \Ex{T} \cdot  4(1 + \ln(2n/\epsilon)/D')$ can be bounded using inequalities from~\cite{janson2017tail}. Denote $\lambda = 4(1 + \ln(2n/\epsilon)/D')$ and observe that $(\lambda - 1)/2 \geq \ln \lambda$ is true since $\lambda > 4$. We get:
	\begin{align*}
	\Pro{T \geq L} &= \Pro{T \geq \Ex{T}\cdot \lambda} \\&\leq \frac{1}{\lambda} \cdot \left(1 - \frac{\ln\Delta}{4 \Delta^{1/\tau}\tau}\right)^{(\lambda - 1 -\ln{\lambda})\Ex{T}} \\&\leq \frac{1}{\lambda} \left(1- \frac{\ln\Delta}{4 \Delta^{1/\tau}\tau}\right)^{\Ex{S}(\lambda - 1)/ 2} \\&\leq \frac{1}{\lambda}e^{-\frac32D' - 2\ln(2n/\epsilon) } \leq \frac{\epsilon^2}{4n^2},
	\end{align*}
	and by taking the union bound over all stations $u$, we get that with probability at least $1-\epsilon^2/(4n)$ the message reaches all nodes within time $4(D + \ln(2n/\epsilon))\cdot  \frac{4 \Delta^{1/\tau} \tau}{\log\Delta}$, conditioned on $S$. Since $S$ takes place with probability at least $1-\epsilon/2$ and since each phase takes $2\tau$ time steps, this shows that the algorithm works within time $8(D + \ln(2n/\epsilon))\cdot  \frac{4 \Delta^{1/\tau} \tau^2}{\log\Delta}$ with probability at least $(1-\epsilon/2)(1-\epsilon^2/(4n)) \geq 1-\epsilon$.
\end{proof}}
\InConference{\vspace*{-2mm}}
\subsection{Lower Bound}
The global broadcast lower bound of
$\Omega(D\log(n/D))$, proved by Kushilevitz and Mansour~\cite{km} for the standard radio network model,
clearly still holds in our setting, as the radio network model is a special case of the dual graph model where $E'=E$.
Similarly, the $\Omega(\log{n}\log{\Delta})$ lower bound proved by Alon~\etal~\cite{alon:1991} also applies.\footnote{This bound is actually
stated as $\Omega(\log^2{n})$, but $\Delta = \Theta(n)$ in the lower bound network, so it can be expressed in terms of $\Delta$ as well for our purposes here.}
It follows that for $\tau \geq \log{\Delta}$,
we almost match the optimal bound for the standard radio network model,
and do match the time of the seminal algorithm of Bar-Yehuda et al.~\cite{Bar-YehudaGI92}. 

For smaller $\tau$,
this performance degrades rapidly.
Here we prove this degradation is near optimal for uniform global broadcast algorithms in our model.
We apply the obvious approach of breaking the problem of global broadcast into multiple sequential instances of local broadcast
(though there are some non-obvious obstacles that arise in implementing this idea). 
As with our local broadcast lower bounds, we separate out the case where $\tau$ is at least a $1/\log\log{\Delta}$ factor smaller than our $\log{\Delta}$ threshold,
as we can obtain a slightly stronger bound under this assumption. 
\InConference{\vspace*{-1mm}}
 \begin{theorem}
 	\label{thm:global_lower}
 		Fix a maximum degree $\Delta \geq 10$, stability factor $\tau$, diameter $D \geq 24$ and uniform global broadcast algorithm $\mathcal{A}$. Assume that $\mathcal{A}$ solves global broadcast in expected time $f(\Delta,D,\tau)$ in all graphs with diameter $D$, maximum degree $\Delta$ and fading adversary with stability $\tau$. It follows that:
 		\begin{enumerate}
 			\item if $\tau < \ln(\Delta-1)/(12\log\log (\Delta-1))$ then $f(\Delta,D,\tau) \in\Omega(D\Delta^{1/\tau}\tau^2/\log\Delta)$,
 			\item if $\tau < \ln(\Delta-1)/16$ then $f(\Delta,D,\tau) \in\Omega(D\Delta^{1/\tau}\tau/\log\Delta)$.
 		\end{enumerate}
 \end{theorem}
\InConference{
\begin{proof}[Proof Idea]
	In this proof we connect together $\Omega(D)$ gadgets used in the proof of Theorem~\ref{thm:lower} (and~\ref{thm:lower_2}) and lower bound the time the message spends in each of the gadgets. The only problem in this approach is that after the message enters to the next gadget, the adversary might not be allowed to change the distribution for some number of steps. We solve this by keeping a distribution that is ``hard'' for the first $\tau$ probabilities of the algorithm in each of the gadgets that has not been reached by the message yet.
\end{proof}
}
\JournalProof{\ProofEight}{
 \begin{proof}
	\begin{figure}
	\centering
	\includegraphics[width=\linewidth]{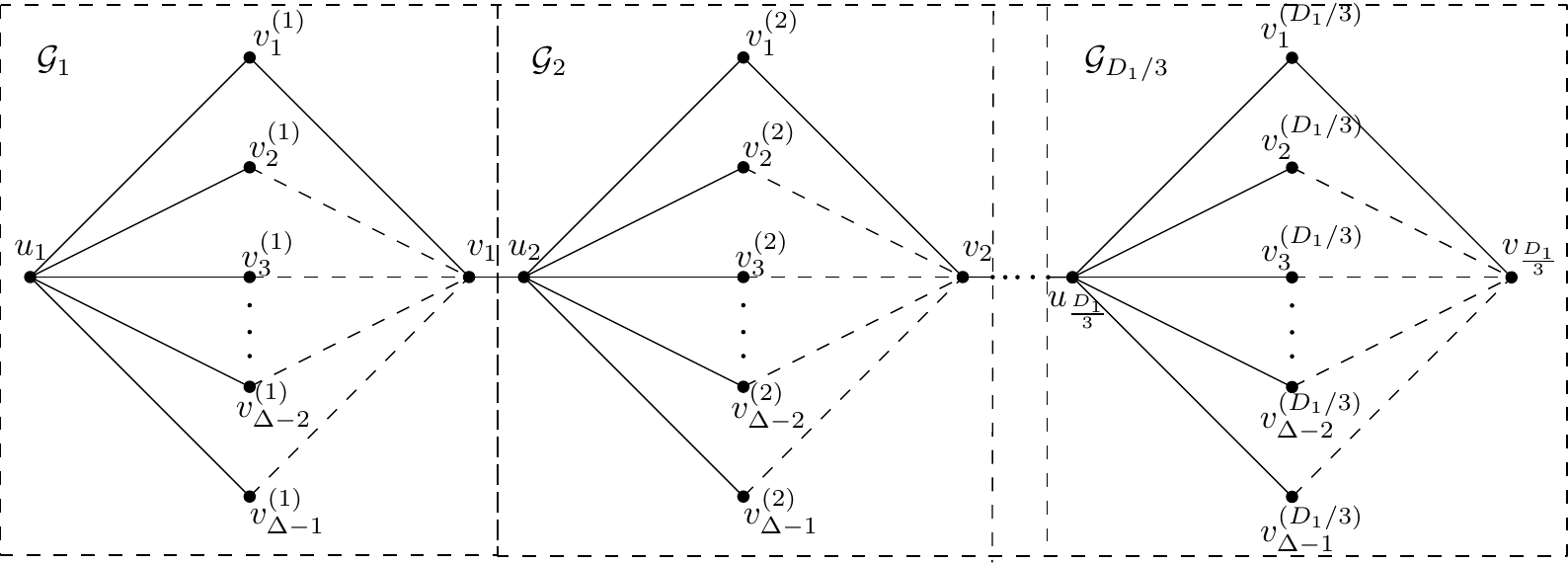}
	\caption{A graph used in proof of Theorem~\ref{thm:global_lower}.}
	\label{fig:lowerGlobal}
\end{figure}
We assume first that $D$ is divisible by $3$ (if it is not we can decrease $D$ by one or two nodes to make it divisible by $3$, without impacting the asymptotic bounds). We construct the dual graph $G, G'$ by connecting together $D/3$ gadgets, $\mathcal{G}_1, \mathcal{G}_2,\dots, \mathcal{G}_{D/3}$, as shown in Figure~\ref{fig:lowerGlobal}. In particular, each gadget $\mathcal{G}_i$ is the same graph structure used to prove our local broadcast lower bound. Formally, for each $i=1,2,…,D/3$, gadget $\mathcal{G}_i$ is a dual graph $G_i = (V_i, E_i)$, $G'_i = (V_i,E'_i)$ where $E_i = \{(u_i,v_{j}^{(i)}) : j =1,2,\dots,\Delta - 1\} \cup \{(v_1^{(i)}, v_i)\}$, $E'_i = E_i \cup \{(v_j^{(i)},v_i),j=2,3,\dots,\Delta  - 1\}$. We denote the set of edges connecting the gadgets by $E_c = \{(v_i,u_{i+1}) : i = 1,2,\dots,D/3 - 1\}$. Finally we can define the total set of nodes and edges in the complete dual graph $G=(V,E)$ and $G'=(V,E')$ as follows: $V = \bigcup_{i=1}^{D/3}V_i$, $E = E_c  \cup \bigcup_{i=1}^{D/3} E_i$, $E' = \bigcup_{i=1}^{D/3} E'_i$. 
We will show statement $1$ by applying Theorem~\ref{thm:lower_2} to each gadget, statement $2$ can be shown using the same proof by applying Theorem~\ref{thm:lower}.

We bound the dissemination of a broadcast message in this graph originating at node $u_1$. We can view the progression of the message through the chain of gadgets $\mathcal{G}_1,\mathcal{G}_2,\dots,\mathcal{G}_{D/3}$ as a sequence of local broadcasts. When the message arrives at a node $u_i$, it is propagated to nodes $v_1^{(i)},v_2^{(i)},\dots,v_{\Delta-1}^{(i)}$  and at this step delivering the message to $v_{i}$ is exactly the local broadcast problem considered in Theorem~\ref{thm:lower_2}. In this theorem we constructed a sequence of distributions that yields a high running time. The distribution changes every exactly $\tau$ steps i.e., we have a distribution $\mathcal{D}_k$ for steps $1 + (k-1)\tau,2+(k-1)\tau,\dots,k\tau$.
 We cannot immediately apply the result for local broadcast because the adversary might not be allowed to change the distribution immediately when the message arrives in a gadget. Moreover in the global broadcast problem, stations are allowed to delay the transmissions for some number of steps. We can easily solve this problem by keeping the ``first" distribution $\mathcal{D}_1$ in each gadget until the message reaches the gadget, at which point the adversary can start the sequence of changes specified by the local broadcast lower bound. 
 
 More precisely, we denote sequence $p_1,p_2,\dots$ of probabilities used by algorithm $\mathcal{A}$ and we denote subsequences $P_k= (p_{1 + (k-1)\tau},p_{2 + (k-1)\tau},\dots,p_{k\tau})$. We want to use distributions $\mathcal{D}_1,\mathcal{D}_2,\dots$ from Theorem~\ref{thm:lower_2} in such a way that if $i$ is the furthest gadget reached by the message and its nodes are in phase $k$ (i.e., are using probabilities from sequence $P_k$) then distribution in gadget $\mathcal{G}_i$ is $\mathcal{D}_k$. If the message has not reach the gadget yet, distribution in the gadget is $\mathcal{D}_1$. Finally if the message already reached node $v_i$ in gadget $\mathcal{G}_i$ for any $i$ we do not change the distribution in this gadget any more.  We need to show that with this construction we do not need to change the distribution more frequently than once per $\tau$ steps. This is true because we only change the distribution in the furthest gadget (call it $\mathcal{G}_i$) reached by the message and moreover we change it from $\mathcal{D}_k$ to $\mathcal{D}_{k+1}$ only after the stations $v_{1}^{(i)},v_{2}^{(i)},\dots,v_{\Delta - 1}^{(i)}$ have finished transmitting with probabilities $P_k$, which takes at least $\tau$ steps (it might take more because stations might delay transmitting with probabilities $P_1$). 

Let us define random variables $X_i$ for $i=1,2,\dots,D/3$ as the number of time steps it takes for nodes $v_1^{(i)},v_2^{(i)},\dots,v_{\Delta - 1}^{(i)}$ to deliver the message to $v_i$. More precisely it is the number of steps between the first step when the stations $v_1^{(i)},v_2^{(i)},\dots,v_{\Delta - 1}^{(i)}$ transmit with probability $p_1$ and the step in which the first successful transmission delivers the message to $v_i$ (including the step of the successful transmission). Note that the steps during which the stations $v_1^{(i)},v_2^{(i)},\dots,v_{\Delta - 1}^{(i)}$  delay transmitting until the beginning of the next probability cycle are not counted in variable $X_i$. The steps counted by variable $X_i$ can be seen as local broadcast. By Theorem~\ref{thm:lower_2} we have that  $\Pro{X_i \leq \DDelta^{1/\tau} \tau^2/(c\ln \DDelta)} \leq 1/2$ for some constant $c > 1$. Moreover variables $X_i$ are independent because choices of the stations in each gadget are independent hence we can use Chernoff bound to lower bound $X = \sum_{i=1}^{D/3}X_i$.  By the assumptions on $D$ we have that $\Ex{X} \geq 8$ (because $\Ex{X_i} \geq 1$), hence
\[
\Pro{X \leq \Ex{X}/2} \leq e^{-\Ex{X}/8}\leq 1/2
\]
Observe that $X$ lower bounds the time of the global broadcast. This shows that the global broadcast needs $\Omega(D\Delta^{1/\tau} \tau^2/\ln \Delta)$ steps with probability at least $1/2$.

If $D$ is not divisible by $3$ we construct our graph with diameter $3\lfloor D/3 \rfloor$ and attach a path of $ D - 3\lfloor D/3 \rfloor$ (one or two) vertices to node $v_{D_1/3}$. This cannot decrease the time of broadcast hence we get the bound $\Omega((D-2) \Delta^{1/\tau} \tau^2 /\log{\Delta}) = \Omega(D \Delta^{1/\tau} \tau^2 /\log{\Delta})$.

 \end{proof}}
\InConference{\vspace*{-2mm}}
\section{Correlations}
Here we explore a promising direction for the study of broadcast in realistic radio network models.
In particular, the fading adversary studied above assumes that the distribution draws are independent.
As we will show,
interesting results are still possible when considering the even more general case where the marginal distributions in each step are not necessarily independent in each round. 
More precisely, in this case, the adversary chooses a distribution over sequences of length at least $\tau$ of the sets of unreliable edges.  A sequence from this distribution is used to determine which unreliable edges are active in successive steps. The adversary after a least $\tau$ steps can decide to change the distribution. In this model, we first show a simple lower bound that any uniform algorithm using a short list of probabilities of length $l$ (our algorithms in previous sections always used list of length $\min\{\tau, \log\Delta \}$) needs time $\Omega(\sqrt{n}/l)$ for some graphs. Our lower bound uses distributions over sequences of graphs in which the degrees of nodes change by a large number in successive steps. Such large changes in degree turn out to be crucial as we show that if in the sequence taken from the distribution chosen by the adversary, in every step in expectancy only $O(\Delta^{1/(\tau - o(\tau))})$ edges adjacent to each node can be changed then we can get an algorithm working in time $O(\Delta^{1/\tau} \tau \log(1/\epsilon))$ with probability at least $1 - \epsilon$ and using list of probabilities of length $O(\min\{\tau,\log\Delta\})$. 

\InConference{\vspace*{-2mm}}
\subsection{A Lower Bound for Correlated Distributions}
\InConference{\vspace*{-1mm}}

The following lower bound shows that any simple back-off algorithm, similar to the ones presented in Section~\ref{sec:local}, that uses at most $\log\Delta$ probabilities requires time $\Omega(\sqrt{\Delta}/\log\Delta)$ if arbitrary correlations are permitted.
\begin{proposition}
	Any uniform local broadcast algorithm that repeats a procedure consisting of $l$ probabilities requires expected time $\Omega(\sqrt{\Delta}/l)$ in some graph with $\Delta = n - 2$ even if $\tau = \infty$. 
\end{proposition}
\InConference{\vspace*{-1mm}}
\begin{proof}
	Denote the procedure that is being used by the algorithm by $\mathcal{P}$. Assume for simplicity that $\sqrt{\Delta}$ is a natural number. We take as a graph a connected pair of stars (a similar graph was used in Theorem~\ref{thm:lower}). 
	
	The fist star has arms $v_1,v_2,\dots,v_{\Delta}$ and center at $u$. In the fist star, arms $v_1,v_2,\dots,v_{\Delta}$ are connected to center $u$ by reliable edges. The second star has arms $v_1,v_2,\dots,v_{\Delta}$ and center at $v$. In the second star, connection from $v_1$ to $v$ is reliable and all other connections are unreliable. Note that by such construction, graph $G$ is connected. All nodes, except $v$, are initially holding a message.
	
		The single distribution is defined in the following way. Let $e_i = \min\{1/p_i, \Delta\}$ for $i=1,2,\dots, l$ be the estimates used by procedure $\mathcal{P}$. Let 
		$$\bar{e}_i = \begin{cases}
		1 \text{ if } e_i \geq \sqrt{\Delta},\\
		n  \text{ if } e_i  < \sqrt{\Delta}.
		\end{cases}$$
		 Let $s$ be a number chosen uniformly at random from $\{1,2,\dots,l\}$. In our distribution, the degree of $v$ in step $t$ is $d_t = \bar{e}_{1 + r_t}$, where $r_t$ is the remainder of $t+s$ modulo $l$. More precisely, in step $t$ in the distribution exactly $d_t - 1$ edges chosen at random among edges between $v$ and $v_2,v_3,\dots,v_{\Delta}$ are activated. Observe that before the algorithm starts, the distribution of the degree of node $v$ in each step is simply a uniform number from multiset $\{\bar{e}_1,\bar{e}_2,\dots,\bar{e}_{l}\}$. But after step $1$ the sequence of degrees of $v$ becomes deterministic and depends only on the value $s$ of the shift. The dependencies are designed in such a way that if $s = l$ (which happens with probability $1/l$) then in any step $t$ of the algorithm, the probability $p_t$ used by the algorithm satisfies either $p_t \cdot d_t \geq \sqrt{\Delta}$ or $p_t \cdot d_t < 1/\sqrt{\Delta}$. This means by Lemma~\ref{lem:prosing} that the success probability is at most $1/\sqrt{\Delta}$ in each step and hence by the union bound the success probability in the whole procedure is at most $l/\sqrt{\Delta}$. Thus with probability at least $1/l$ the algorithm has to repeat procedure $\mathcal{P}$ at least $\sqrt{\Delta}/(2l)$ times to get a constant probability of success. Hence the expected time is $\Omega(\sqrt{\Delta}/l)$.
\end{proof}
\InConference{\vspace*{-2mm}}
\subsection{Locally Limited Changes}
The previous section shows that under an adversary that is allowed to use arbitrary correlations then any simple procedure need polynomial time in the worst case. 

In this section we want to consider the adversary that can use correlations but cannot change the degree too much in successive steps. Of course once every at most $\tau$ steps the adversary is allowed to define a completely new distribution over the unreliable edges. We want to argue that it is possible to build a simple algorithm resistant to such an adversary. Intuitively the changes of the degree are problematic only if the changes are by a large (non-constant) factor. Note by Lemma~\ref{lem:prosing} that if we perturb the effective degree by only a constant factor then the bound also changes only by a constant factor.  Hence in order to design an algorithm that is immune to such changes we should add more ``coverage'' to the small-degree nodes. We do this by enhancing each phase of algorithm $\algnameone$ with additional steps in which we assume that the effective degree of a node is small. The adversary may try to avoid the successful transmission in these steps by changing the degree (the adversary knows the probabilities used by the algorithm). But having the restriction on the distance the adversary can move the degree allows us to define overlapping ``zones'' such that in two consecutive steps we are sure to find the degree in one of the zones. We also have to make sure that the whole phase of the new algorithm fits into $\tau$ steps.

Now we present algorithm $\algnamethree$ (Robust Local Broadcast with Correlations). We first show that the algorithm works under $(l,\tau)$--deterministic adversary that can change at most $l$ edges adjacent to each node per round and all the edges from $E' \setminus E$ once every at most $\tau$ rounds. Our algorithm will be resistant to deterministic adversary that can change at most $\tau\Delta^{1/(\tau- o(\tau))}$ edges adjacent to each node in every step. 

Then we show that it also works under restricted fading adversary with parameters $\tau$ and $l$. Restricted fading adversary can change the distribution arbitrarily once every at most $\tau$ steps, if the distribution is not changed then the expected change of the degree of any node can be at most $l$. Under these restrictions, the adversary can design arbitrary correlations between successive steps. We show that $\algnamethree$ works with restricted fading adversary with $l$ of at most $\Delta^{1/(\tau- o(\tau))}$. 
\InConference{\vspace*{-1mm}}
\begin{algorithm}[h]
	\TitleOfAlgo{\algnamethree$(r,\tau)$}
	$\bar{\tau} = \min\{ \lceil \log_{2e}\Delta / 2 \rceil, \tau \}$ \\
	$a \gets \lceil\bar{\tau}/\log_{2e}{\bar{\tau}}\rceil$ \\
	$k \gets \lceil \Delta^{1/(\tau -2a)} \rceil$ \\ 
    $e_1 \gets k \cdot  a$ \\
    $e_2 \gets k^2 \cdot \tau \cdot a $ \\
	\Repeat($2r$ times){
		\algnameone$(1,\bar{\tau} -2 a )$\\
		\Repeat($a$ times){ 
	     \procname$(1,1/e_1)$ \\
		\procname$(1,1/e_2)$ }
	}
\end{algorithm}
\InConference{\vspace*{-1mm}}
\begin{theorem}
	\label{thm:corr_deter}
	If $\tau \geq 1000$ Algorithm \algnamethree$( 8 e \lceil \ln(1/\epsilon) \Delta^{1/\tau} \rceil ,\tau)$ solves local broadcast in the presence of $\left( \left\lfloor\Delta^{\frac{1}{\tau- 2\lceil \tau/\log_{2e}\tau\rceil }} \right\rfloor \tau / 2,\tau\right)$-deterministic adversary in time $O(\Delta^{1/\tau} \tau \log(1/\epsilon))$ with probability at least $1-\epsilon$.
\end{theorem}
\InConference{
	\begin{proof}[Proof Idea]
		For a fixed receiver $v$ we want to show that the probability that $v$ receives the message in one of the $r$ cycles (each $2$ iterations of loop in Lines $7-11$ is one cycle) is at least $p_s = \frac{1}{8e k}$. We do it by separately considering two cases depending on degree $d_t(v)$, where $t$ is the first step of the considered cycle. If $d_t(v) \geq 2 l^2$ we can show that the degree cannot change in total in this cycle by more than a factor of $2$ (here we use the restriction on the adversary) in which case we can show that in one of the steps of procedure $\algnameone$ the probability of success is at least $p_s$. For smaller degrees $d_t(v) < 2l^2$ we pick $a$ pairs of steps such that in the first step of the pair the algorithm uses probability $1/e_1$ and in the second it uses $1/e_2$. Then we observe that either in the first step of the pair the degree is at most $2l$ in which case broadcasting with probability $1/e_1$ gives probability $p_s / a$ of success. In the opposite case the degree is at least $l$ (here we use the restriction on the adversary) in the second step and broadcasting with probability $1/e_2$ gives probability $p_s / a$ of success. Since we have $a$ such pairs the claim follows.
	\end{proof}
}
\JournalProof{\ProofTen}{
\begin{proof}
	Assume that $\tau \leq \lceil \log_{2e}\Delta / 2 \rceil$ and note that in this case $\bar{\tau}  = \tau$. In the opposite case we use the algorithm for $\tau = \lceil \log_{2e}\Delta / 2 \rceil $ which works also for any larger $\tau$. Denote $k = \lfloor \Delta^{1/(\tau - 2a)} \rfloor$, $l = k\tau /2$ and observe that for $\tau \geq 1000$ we have $a > 200$ and $\tau - 2a \geq \tau /2$ and $k \geq 2$. We divide the time into intervals of length $\tau$, called \textit{cycles}. In each interval algorithm $\algnamethree$ repeats the same probabilities. In the first $\tau -2a$ steps of the cycle it uses probabilities $p_i = k^{-i}$ for $i = 1,2,\dots,\tau -2a$, in the next $2a$ steps it uses probabilities $1/e_1$ and $1/e_2$.  We take two consecutive cycles and note that in each such a pair of cycles we can find $\tau$ consecutive steps in which the distribution over the unreliable edges is the same (since global changes can happen at most once every $\tau$ steps) and moreover the algorithm uses all the probabilities from a cycle. Let us call a sequence of these steps $T = [t_1,\dots,t_{\tau}]$. Note that in this sequence we have either one full procedure \algnameone$(1,\tau - 2a)$ or parts of two procedures \algnameone$(1,\tau - 2a)$  (call them $R_1$ and $R_2$). In the second case sequence $T$ contains some suffix of $R_1$ and some prefix of $R_2$. Connect these steps together into a procedure $R$, which contains all steps of procedure \algnameone$(1,\tau-2a)$ executed in a possibly different order. 
	Fix a receiver $v$ and assume that at least one reliable neighbor of $v$ tries to transmit a message to $v$. We want to show that in each such a pair of cycles $v$ receives the message independently with probability at least $p_s = \frac{1}{8e k}$.
	
	We know by the definition of the adversary that the effective degree cannot change by too much between steps in the same cycle: $|d_{t_i}(v) - d_{t_{i} + 1}(v)| \leq l$. We can consider two cases depending on the effective degree in the first considered step $t_1$:
	\subparagraph{Case 1:} $d_{t_1}(v) \geq 2 l^2$ \\
	Here we want to show that procedure $R$ is successful with probability at least $p_s$.
	Observe that here since $l \geq \tau$ we have $d_{t_i}(v) \geq d_{t_i}(v)/2 +   l^2  \geq d_{t_i}(v) /2 + l\tau $ for each $i = 1,2,\dots,\tau$. Thus $d_{t_i}(v) - l\tau \geq d_{t_i}(v)/2$  and $d_{t_i}(v) +l\tau \leq  2 d_{t_i}(v)$ thus the effective degree in the whole considered sequence of steps can change by a factor of at most $2$. Recall from the definition of $\algnameone$ that It uses probabilities $p_i = k^{-i}$. Consider the smallest $i$ such that $1/p_i \geq 2 d_{t_1}(v)$ by the minimality of $i$ we have that $1/(kp_i) \leq 2 d_{t_1}(v)$. Probability $p_i$ is used in some step of sequence $T$. Call this step $t_j$. We have:
	\[
	1/p_i \geq 2d_{t_1}(v) \geq d_j(v) \geq d_{t_1}(v)/2 \geq 1/(4kp_i).
	\]  
	Thus by Lemmas~\ref{lem:prosing} and~\ref{lem:interval}:
	\[
	\Pro{R_{t_j}^{(v)}} = \min\left\{ (2e)^{-1}, \frac{(2e)^{-1/(4k)}}{4k} \right\} \geq \frac{1}{8ek} = p_s.
	\]
	\subparagraph{Case 2:}  $d_{t_1}(v) <2 l^2 $ \\
	Here we want to show that a successful transmission occurs with probability at least $p_s$ in one of the $2a$ additional steps (see lines $7-11$ of the pseudocode). 
	
	Note that since $d_1(v) < 2l^2$ then $d_{t_i}(v) \leq d_{t_1}(v) + l\tau \leq 4 l^2$
	Pick two consecutive steps $t_i,t_i+1$ such that in step $t_i$ the algorithm uses probability $1/e_1$ and in $t_i+1$ it uses $1/e_2$. Note that in the considered sequence we have at least $a-1$ such pairs. 
	\subparagraph{Case 2.1: $d_{t_i}(v) \leq 2l$}
	Here the probability is $1/e_1$ and the degree is within interval $[1,2l]$ hence we have that:
	\[
   \frac{\tau}{a} = \frac{2l}{e_1} \geq \frac{d_{t_i}}{e_1} \geq \frac{1}{e_1}.
	\]
	By Lemma~\ref{lem:prosing}~\ref{lem:interval}:
	\begin{align*}
	\Pro{R_{t_i}^{(v)}} & \geq \min\left\{\frac{2l}{e_1}e^{-2l/e_1},  1/e_1 e^{-1/e_1}\right\}\geq \min\left\{\frac{\ln\tau}{\tau}, \frac{1}{eka}\right\} \geq \frac{1}{eka}.
	\end{align*}
	\subparagraph{Case 2.2: $4l^2 \geq d_{t_i}(v) > 2l$}
	Note that in this case $d_{t_{i}+1} (v)  \in [l,4l^2]$ and the probability used in this step is $1/e_2$ hence:
		\[
	\frac{\tau}{a} =\frac{4l^2}{e_2} \geq \frac{d_{i+1}(v)}{e_2} \geq \frac{2l}{e_2} = \frac{1}{e_1},
	\]
	and we can use Lemmas~\ref{lem:prosing} and~\ref{lem:interval} exactly the same as in the previous case and obtain $\Pro{R_{t_i +1}^{(v)}} \geq \frac{1}{eka}$.
	
	In each pair the stations are making independent choices hence the probability of failure in all the pairs is by Lemma~\ref{lem:wpi} at most:
	
	\[
	\left(1-\frac{1}{eka}\right)^{a-1} \leq 1 - \frac{a-1}{eka} + {a-1 \choose 2}\frac{1}{e^2 k^2 a^2} \leq 1-\frac{1}{2ek},
	\]
	where in the last inequality we used the fact that $a > 20$. Thus also in this case with probability at least $1/(2ek) \geq p_s$ node $v$ receives a message during this cycle. 
	
	The two considered cases showed that any full two cycles deliver the message with probability at least $p_s$. If we perform at least $2r = 2\lceil \ln(1/\epsilon)/p_s \rceil = O(\Delta^{1/\tau})$ cycles then the probability that $v$ does not receive a message is at most 
	$
	(1-p_s)^{\ln(1/\epsilon)/p_s} \leq \epsilon.
	$
\end{proof}}
\InConference{\vspace*{-1mm}}
The case with deterministic adversary can be generalized to stochastic restricted adversary.
\begin{theorem}
	\label{thm:corr_random}
		If $\tau \geq 1000$ Algorithm \algnamethree$(  16 e \lceil \ln(1/\epsilon) \Delta^{1/\tau} \rceil ,\tau)$ solves local broadcast in the presence of 
		$l$-restricted fading adversary using correlations with $l = \left\lfloor \Delta^{\frac{1}{\tau (1-1/\log_{2e}\tau)}} \right\rfloor / 4$ in time $O(\Delta^{1/\tau} \tau \log(1/\epsilon))$ with probability at least $1-\epsilon$.
\end{theorem}
\InConference{\vspace*{-1mm}}
\InConference{
	\begin{proof}[Proof Idea]
		We show that if an algorithm works with $2l\tau$-deterministic adversary then it also works with $l$-stochastic adversary with correlations. We note that by Markov's inequality with probability at least $1/(2\tau)$ the degree of the receiver changes by at most $2l\tau$. By the union bound with probability at least $1/2$, the degree does not change by more then $2l\tau$ throughout the whole cycle of length $\tau$. For such cycles, the analysis of the deterministic case gives us probability $p_s$ of success. Thus in the stochastic case the probability of success in each cycle is at least $p_s/2$.
	\end{proof}
}
\JournalProof{\ProofEleven}{
\begin{proof}
	Fix any receiver $v$. We know that \algnamethree$(  8 e \lceil \ln(1/\epsilon) \Delta^{1/\tau} \rceil ,\tau)$ solves local broadcast in the presence of $(l\tau,\tau)$ - deterministic adversary. But in the case with arbitrary correlations we can still bound the probability that the degree of $v$ does not change too much. Take any two consecutive steps $t,t+1$. We have by Markov Inequality:
	\[
	\Pro{|d_{t}(v) - d_{t+1}(v)| > 2 \tau l} \leq 1/(2 \tau) 
	\]
	If we pick $\tau$ steps like in the proof of Theorem~\ref{thm:corr_deter} then by the union bound with probability at least $1/2$ in each of these steps the degree changes by at most $2l\tau$. From now on we can use the same analysis as in Theorem~\ref{thm:corr_deter} and we obtain only a constant slowdown compared to the case with deterministic adversary. Hence \algnamethree$(  16 e \lceil \ln(1/\epsilon) \Delta^{1/\tau} \rceil ,\tau)$ solves local broadcast with restricted fading adversary with probability at least $1-\epsilon$.
\end{proof}
}

\paragraph{Acknowledgements}
Many thanks to William Kuszmaul and Zachary Newman for helpful comments on an earlier version of this manuscript.

\InConference{\vspace*{-1mm}}
\bibliographystyle{plainurl}
\bibliography{Local}

\begin{thebibliography}{10}

\bibitem{alon:1991}
N.~Alon, A.~Bar-Noy, N.~Linial, and D.~Peleg.
\newblock {A Lower Bound for Radio Broadcast}.
\newblock {\em Journal of Computer and System Sciences}, 43(2):290--298, 1991.

\bibitem{Bar-YehudaGI92}
Reuven Bar{-}Yehuda, Oded Goldreich, and Alon Itai.
\newblock On the time-complexity of broadcast in multi-hop radio networks: An
  exponential gap between determinism and randomization.
\newblock {\em J. Comput. Syst. Sci.}, 45(1):104--126, 1992.
\newblock \href {http://dx.doi.org/10.1016/0022-0000(92)90042-H}
  {\path{doi:10.1016/0022-0000(92)90042-H}}.

\bibitem{CGKLN-jour}
Keren Censor-Hillel, Seth Gilbert, Fabian Kuhn, Nancy Lynch, and Calvin
  Newport.
\newblock Structuring unreliable radio networks.
\newblock {\em Distributed Computing}, 27(1):1--19, 2014.

\bibitem{ClementiMS04}
Andrea E.~F. Clementi, Angelo Monti, and Riccardo Silvestri.
\newblock Round robin is optimal for fault-tolerant broadcasting on wireless
  networks.
\newblock {\em J. Parallel Distrib. Comput.}, 64(1):89--96, 2004.

\bibitem{Ghaffari-MS}
Mohsen Ghaffari.
\newblock Bounds on contention management in radio networks.
\newblock Master's thesis, Electrical Engineering and Computer Science,
  Massachusetts Institute of Technology, Cambridge, MA, February 2013.

\bibitem{GHLN-disc12}
Mohsen Ghaffari, Bernhard Haeupler, Nancy Lynch, and Calvin Newport.
\newblock Bounds on contention management in radio networks.
\newblock In Marcos~K. Aguilera, editor, {\em Distributed Computing: 26th
  International Symposium (DISC 2012), Salvador, Brazil, October, 2012}, volume
  7611 of {\em Lecture Notes in Computer Science}, pages 223--237. Springer,
  2012.

\bibitem{GKLN14}
Mohsen Ghaffari, Erez Kantor, Nancy Lynch, and Calvin Newport.
\newblock Multi-message broadcast with {A}bstract {MAC} layers and unreliable
  links.
\newblock In {\em Proceedings of the 33nd Annual ACM Symposium on Principles of
  Distributed Computing (PODC'14)}, pages 56--65, Paris, France, July 2014.

\bibitem{GhaffariLynchNewport-sub}
Mohsen Ghaffari, Nancy Lynch, and Calvin Newport.
\newblock The cost of radio network broadcast for different models of
  unreliable links.
\newblock In {\em Proceedings of the 32nd Annual ACM Symposium on Principles of
  Distributed Computing}, pages 345--354, Montreal, Canada, July 2013.

\bibitem{ghaffari:2016}
Mohsen Ghaffari and Calvin Newport.
\newblock A leader election in unreliable radio networks.
\newblock In {\em Proceedings of the International Colloquium on Automata,
  Languages, and Programming (ICALP)}, 2016.

\bibitem{ba}
Seth Gilbert, Nancy~A. Lynch, Calvin Newport, and Dominik Pajak.
\newblock Brief announcement: On simple back-off in unreliable radio networks.
\newblock In {\em 32nd International Symposium on Distributed Computing, {DISC}
  2018, New Orleans, LA, USA, October 15-19, 2018}, pages 48:1--48:3, 2018.

\bibitem{janson2017tail}
Svante Janson.
\newblock Tail bounds for sums of geometric and exponential variables.
\newblock {\em arXiv preprint arXiv:1709.08157}, 2017.

\bibitem{KLN-podc09}
Fabian Kuhn, Nancy Lynch, and Calvin Newport.
\newblock Brief announcement: Hardness of broadcasting in wireless networks
  with unreliable communication.
\newblock In {\em Proceedings of the 28th Annual ACM Symposium on the
  Principles of Distributed Computing (PODC 2009)}, Calgary, Alberta, Canada,
  August 2009.

\bibitem{KLNOR}
Fabian Kuhn, Nancy Lynch, Calvin Newport, Rotem Oshman, and Andrea Richa.
\newblock Broadcasting in unreliable radio networks.
\newblock In {\em Proceedings of the 29th ACM Symposium on Principles of
  Distributed Computing (PODC)}, pages 336--345, Zurich, Switzerland, July
  2010.

\bibitem{km}
E.~Kushilevitz and Y.~Mansour.
\newblock {An $\Omega$(D log(N/D)) Lower Bound for Broadcast in Radio
  Networks}.
\newblock {\em SIAM Journal on Computing}, 27(3):702--712, 1998.

\bibitem{lynch:2015}
Nancy Lynch and Calvin Newport.
\newblock A (truly) local broadcast layer for unreliable radio networks.
\newblock In {\em Proceedings of the ACM Symposium on Principles of Distributed
  Computing (PODC)}, 2015.

\bibitem{newport:2014b}
Calvin Newport.
\newblock Lower bounds for radio networks made easy.
\newblock In {\em Proceedings of the International Symposium on Distributed
  Computing (DISC)}, 2014.

\bibitem{fading}
Mike Willis.
\newblock Propagation tutorial: Fading.
\newblock \url{http://www.mike-willis.com/Tutorial/PF15.htm}.
\newblock May 5, 2007.

\end{thebibliography}

%
%
%
%
%
%
%

\end{document}